\documentclass[a4paper]{article}
\usepackage{typearea}
\usepackage{graphicx}  
\usepackage{hyperref}
\usepackage{makeidx}
\usepackage{amsmath}

\usepackage{mathtools}
\usepackage{amssymb}
\usepackage{amsthm}
\usepackage{mathrsfs}
\usepackage{color}
\usepackage{bm}
\usepackage{booktabs}
\usepackage{ytableau}
\usepackage{cite} 
\DeclareSymbolFont{extraup}{U}{zavm}{m}{n}
\DeclareMathSymbol{\varheart}{\mathalpha}{extraup}{86}
\DeclareMathSymbol{\vardiamond}{\mathalpha}{extraup}{87}

\include{settings}
\numberwithin{equation}{subsection}
\begin{document}
\typearea{18}
\title{Green's function in general relativity}
\author{Yoshimasa Kurihara\footnote{yoshimasa.kurihara@kek.jp}\\
{\small 
High Energy Accelerator Organization (KEK), 
Tsukuba, Ibaraki 305-0801, Japan}
}
\date{\today}
\maketitle

\begin{abstract} 
This report provides Green's functions (classical propagators) of gravitational fields appearing in general relativity.
The existence of Green's function of the wave equation in curved space with an indefinite metric is ensured owing to the Hodge harmonic analysis.
The analyticity of Green's function is determined intrinsically to keep a causality.
This report proposed a novel definition of momentum space in curved space-time and the linearisation of the Einstein equation as a free field consistent with that of the Yang--Mills gauge field.
The proposed linearisation does not utilize the weak-field approximation; thus, the method applies to highly caved space-time.
We gave two examples of Green's function of gravitational fields: the plane wave solution and the Schwarzschild solution.
\end{abstract}

 \tableofcontents
  
\section{Introduction}
General relativity treats the gravitational force classically and is supported by various experiments, including recent observations of the gravitational wave
General relativity is also the Yang-Mills type gauge theory with the global $GL(4)$ and the local $SO(1,3)$ gauge symmetry, as first Utiyama has clarified\cite{PhysRev.101.1597}.
A perturbative renormalisation theory of gravity has not yet been established, in contrast to the Yang--Mills theory with the compact Lie group.
Although various attempts have been made for nearly one hundred years, a satisfactory theory of quantum gravity still needs to be completed.
In the absence of full quantum gravity, studying quantum field theory in curved space-time provides a tool to investigate the quantum effect of gravity.
Fundamental and exciting phenomena, e.g., the Hawking radiation\cite{1974Natur.248...30H}, the Unruh effect\cite{PhysRevD.14.870} and others, are found owing to the quantum field theory in curved space-time.
The progress and results of studies are summarised in many books and lectures\cite{Birrell:1982ix,Wald:1995yp,Ford:1997hb,mukhanov2007introduction,61354,Parker:2009uva}.

In 1956, Utiyama\cite{PhysRev.101.1597} first pointed out that general relativity is a Yang-- Mills-type gauge theory\cite{PhysRev.96.191}.
We refer to the Yang--Mills theory, including general relativity as the local $SO(1,3)$ gauge theory, as the \textit{Yang--Mills-Utiyama} (\YMU\!) \textit{theory}.
According to this point of view, vierbein and spin-connection fields are natural counterparts of matter and gauge fields\cite{Kurihara_2020} in the \YMU theory.
The current author has provided a full geometrical Lagrangian of the \YMU theory\cite{Kurihara:2025tro} and proposed the canonical quantisation of Einstein's gravity\cite{doi:10.1140/epjp/s13360-021-01463-3} as the gauge theory of vierbein and spin-connection. 
In that study, the author quantised the vierbein field using the Heisenberg picture without a weak-field approximation around a flat metric and did not apply a perturbation approach; thus, perturbative propagators of vierbein and spin-connection fields are not given there.
This report's primary objective is to provide both perturbative and exact Green's functions of vierbein and spin-connection fields.
In addition, we propose a novel linearisation of the Einstein equation for curved space-time for perturbation.

Green's function, given in the momentum space obtained using the Fourier transformation, is the propagator of quantised fields.
An integration kernel of the Fourier transformation is the fundamental solution of the wave equation in the flat Minkowski space.
The wave equation is the equation of motion for the quantised free field, and its solution forms an orthonormal system; thus, the propagator provides the perturbative solution of the interacting equation of motion expanded by the free-field solutions.
On the other hand, propagators of the Yang-Mills gauge and spinor fields are not trivial in curved space-time.
When background space-time is not a flat Minkowski space, an equation of motion for the free (no gauge-interacting) field is not the wave equation, and the physical meaning of the momentum space is unclear.
The second objective of this report is to clarify the physical meaning of the Fourier transformation and the momentum space in curved space-time. 
Moreover, we also discuss the existence of Green's function in Minkowski space-time. 

The structure of this report is as follows:
After this introduction, we present mathematical preparations of inertial space-time and extended Riemannian space, namely the amphometric space,  in \textbf{section 2}.
In \textbf{section 3}, we discuss the existence of Green's function in an indefinite metric space.
In curved space-time, the Fourier transformation and momentum space are not trivial and have a different structure from those in the flat Minkowski space.
In \textbf{section 4}, we define the Fourier transformation, configuration space, and momentum space in curved space-time.
\textbf{Section 5} provides the main results of this report, perturbative and exact Green's functions of vierbein and spin-connection fields.
Finally, \textbf{section 6} summarises results and discussions.

We use the following physical units in this study: values of speed of light $c$, reduced Plank-constant $\hbar=h/2\pi$, and Einstein gravitational constant $\kappa$ are set to unity $c=\hbar=\kappa=1$, whereas $\hbar$ and $\kappa$ are written explicitly in formulae to clarify the physical dimension of variables.
This unit system provides physical dimensions of fundamental constants as $\fdl\hbar\fdr=LE=TM$, $\fdl\hbar\hspace{.1em}\kappa\fdr=L^2=T^2$ and $\fdl\hbar/\kappa\fdr=E^2=M^2$, where $L$, $T$, $E$ and $M$ are, respectively, length, time, energy and mass dimensions.

We treat both $SO(1,3)$ and $SO(4)$ symmetries simultaneously using the one-parameter family of metric tensors, namely the \textit{amphometric}\cite{Kurihara:2025tro}.
We denote a four-dimensional special-orthogonal group as $SO_4$ when we do not distinguish its signature.

\section{Inertial space with amphometric}\label{intro}
This section summarises the mathematical tools utilised throughout this study.
%
%
\subsection{Inertial bundle}\label{Inertial bundle}
We introduce four-dimensional Riemannian manifold $(\MM,\bm{g})$, where $\MM$ is a smooth and oriented four-dimensional manifold, namely the \emph{global space-time manifold}, and $\bm{g}$ is a metric tensor in $\MM$.
Metric tensor $\bm{g}$ is the non-degenerate symmetric tensor; thus, it has four eigenvectors.
We assume $\bm{g}$ has one positive and three negative eigenvalues. 
In an open neighbourhood around $p$, $U_{\hspace{-.1em}p{\in}\MM}\subset\MM$, we introduce the standard coordinate $x^\mu$; orthonormal bases in $T\MM_{\hspace{-.1em}p}$ and $T^*\MM_{\hspace{-.1em}p}$ are represented $\partial/\partial x^\mu$ and $dx^\mu$ using the standard coordinate, respectively.
We use the abbreviation $\partial_\mu:=\partial/\partial x^\mu$ throughout this report.
They are dual bases to each other such that $dx^\mu\cdot\partial_\nu=\delta^\mu_\nu$, where a \emph{dot product} $\xxx{\cdot}y$ means a contraction of a form $\xxx$ in $\TsMM_{\hspace{-.1em}p}$ concerning a vector $y$ in $\TMM_{\hspace{-.1em}p}$.
Two trivial vector bundles $\TMM:=\bigcup_{p} T\MM_{\hspace{-.1em}p}$ and $\TsMM:=\bigcup_{\hspace{-.1em}p} \TsMM_{\hspace{-.1em}p}$ are referred to as  a tangent and cotangent bundles in $\MM$, respectively.

The Levi-Civita connection in $\MM$ is defined using a metric tensor as
\begin{align*}
\Gamma^\lambda_{~\mu\nu}(x):=\frac{1}{2}g^{\lambda\sigma}(x)\(
\partial_\mu g_{\nu\sigma}(x)+\partial_\nu g_{\mu\sigma}(x)-\partial_\sigma g_{\mu\nu}(x)
\),
\end{align*}
where $g_{\mu\nu}:=[\bm{g}]_{\mu\nu}$.
An inertial system in which the Levi-Civita connection vanishes such that $\Gamma^\lambda_{~\mu\nu}(x)=0$ exists locally at any point $x\in\MM$.
We define an inertial system more precisely as follows:
\begin{definition}\label{InrSys}\textup{(Inertial system\cite{fre2012gravity})}\\
When a coordinate system $\xi^\mu$ in $U_{\hspace{-.1em}p}$ fulfils the following conditions, $\xi^\mu$ is called the inertial system: 
\begin{enumerate}
\item An origin of $\xi^\mu$ is at $p$ such that $\xi^\mu(x)\big|_{x{\rightarrow}p}=0$.
\item The metric tensor concerning the $\xi^\mu$ coordinate at $p$ is ${\bm{\eta}}(x)\big|_{x{\rightarrow}p}=\textup{diag}[1,-1,-1,-1]$.\label{2}
\item The Levi--Civita connection concerning the $\xi^\mu$ coordinate at $p$ vanishes such that $\Gamma^\lambda_{~\mu\nu}(x)\big|_{x{\rightarrow}p}=0$.\QED
\end{enumerate}
\end{definition}
\noindent
The existence of the inertial system at any point in $\MM$ is referred to as \textit{Einstein's equivalence principle} in physics.
A manifold equipping an inertial frame at point $p\in\MM$ is denoted as $\M_{p}$, namely ``a local inertial manifold at a point $p$''.
A four-dimensional \textit{rotation} (the Lorenz transformation) is isometry $\Lambda_{p}:\text{End}(\M_{p})$ concerning a local $SO(1,3)$ group.
A principal inertial bundle is defined as a tuple such that:
\[
\(\M,\pi_\textsc{i},\MM,SO(1,3)\).
\]
A total space $\M$, namely an \emph{inertial manifold}, is a trivial bundle of quotient spaces such that:
\begin{align}
\M:=\bigcup_{p\in\M}\M_{p}/\Lambda_{p},\label{Lmani}
\end{align}
and the bundle map is a smooth map such that: 
\begin{align*}
\pi_\textsc{i}:\M\otimes\MM\rightarrow\MM:\M_{p}\mapsto{p}.
\end{align*}
An orthonormal basis on $\TM_{p}$ is represented as 
\begin{align*}
\partial_a:=\frac{\partial~}{\partial\xi^a(p)}\in V(\TM_p),
\end{align*}
where $\xi^a$ is the standard basis at $p$ in the open neighbourhood $U_{p{\in}\M}$.
In other words, $\xi^a(p)$ is a function mapping a point $p\in\M$ to a frame vector $\xi^a$ at $p$, and frame vector $x^\mu(p)$ is a function mapping a point $p\in\MM$ to the vector $x^\mu$.
We also use abbreviations such as
\begin{align*}
x^\mu(\xi)&:=x^\mu\(\pi_p\(\xi^{-1}\(p\in\M\)\)\)\in\V(\TMM_p),\quad
\xi^a(x):=\xi^a\(\pi_p^{-1}\(x^{-1}\(p\in\MM\)\)\)\in\V(\TM_p),
\intertext{where}
\pi_{p}&:=p\in\M{\mapsto}p\in\MM\!.
\end{align*}

As for suffixes of vectors in $\TM_{p}$, Roman letters are used for components of the standard basis throughout this study; Greek letters are used for them in $\MM_{p}$.
Owing to the convention, we can distinguish two abbreviated vectors such that $\partial_\mu\in V(\TMM)$ and $\partial_a\in V(\TM)$.
The metric tensor in $\TsM_{p}$ is $\bm{\eta}$ in \textbf{Definition  \ref{InrSys}}-\textit{\ref{2}}.
The Levi-Civita tensor (complete anti-symmetric tensor) $\bm\epsilon$, whose component is $[\bm\epsilon]_{0123}=\epsilon_{0123}=+1$, and $\bm\eta$ are constant tensors in $\TsM:=\bigcup_{p\in\M}\TsM_p$.
Later in this study, they are extended to the amphometric.

A pull-back of the bundle map induces a one-form object $\eee^a\in\Gamma(\TsM_{p},\M)$ represented using the standard basis as
\begin{align}
\pi_\textsc{i}^\#&:
\Omega^1(\TsMM_{\hspace{-.2em}p})\rightarrow\Gamma(\TsM_{p},\M):
dx^\mu\mapsto
\eee^a\(\xi(p)\):=d\xi^a(p),\label{pull}
\end{align}
where $\Omega^n(\TsMM_{\hspace{-.1em}p})$ is a space of differential $n$-forms in $\TsMM_{\hspace{-.1em}p}$ and $\Gamma$ is a space of sections.
We define the vierbein $\E^a_\mu(x)\in{C^\infty(\MM)}$ as a matrix representation of the map (\ref{pull}) such that:
\begin{align*}
\eee^a\(\xi(p)\)=
\pi_\textsc{i}^\#\(dx^\mu(p)\)&:=
\E^a_\mu(x)dx^\mu(p\in\MM_p)
=d\xi^a(p\in\M_p). 
\end{align*}
The vierbein is a smooth and invertible function globally defined in $\MM$.
Throughout this report, we use a Fraktur letter to represent a $p$-form object defined in the cotangent bundle\footnote{A Fraktur letter is also used for Lie-algebra.}.
The Einstein convention for repeated indexes (one in up and one in down) is exploited throughout this study.

The vierbein inverse $[\E^{-1}]_a^\mu=\E_a^\mu(\xi)\in{C^\infty(\M)}$, which is also called the vierbein, is an inverse transformation such that:
\begin{align*}
\E^{-1}&:\Omega^1(\TsM_p)\rightarrow\Omega^1(\TsMM_p):
\eee^a(\xi(p))\mapsto dx^\mu(p)=\E_a^\mu(\xi)\eee^a\(\xi\)\(p\in\MM\),
\end{align*}
yielding
\begin{align*}
\E^a_\mu(x)\E_a^\nu(\xi(x))&:=\delta^\nu_\mu
\quad\text{and}\quad
\E_a^\mu(\xi)\E^b_\mu(x(\xi)):=\delta_a^b.
\end{align*}
Owing to Definition \ref{InrSys}-\textit{\ref{2}}, metric tensors in $\MM$ and those in $\M$ are related to each other as
\begin{align*}
g^{~}_{\mu\nu}(x)&=\E^\bcdot_\mu(x)\E^\bcdot_\nu(x)\hspace{.1em}\eta^{~}_\bcdots(\xi(x))
\quad\text{and}\quad
\eta^{~}_{ab}(\xi)=\E_a^\mu(\xi)\E_b^\nu(\xi)\hspace{.1em}g^{~}_{\mu\nu}(x(\xi)).
\end{align*}
Dummy Roman indexes are often abbreviated to a small circle $\bcdot$ (or $\star$) when the dummy-index pair of the Einstein convention is obvious as above.
When multiple circles appear in an expression, the pairing must be on a left-to-right order at upper and lower indexes.
The same map yields a push-forward of a vector in $\TM$ for $\TMM$ such that: 
\begin{align*}
{\pi_\textsc{i}^\#}^{-1}=:{\pi_\textsc{i}}_{\#}&:
V^1\(\TM_p\){\rightarrow}V^1\(\TMM_p\):
\frac{\partial~}{\partial\xi^a(p)}\mapsto
\frac{\partial~}{\partial x^\mu}:=\E^a_\mu(x)
\frac{\partial~}{\partial \xi^a(x)}.
\end{align*}
We use simplified representations, e.g., $\E_a^\mu$ and $\E^a_\mu$ instead of $\E_a^\mu(\xi)$ and $\E^a_\mu(x)$ when its domain is obvious owing to their suffix.

One-form object $\eee^a$ is referred to as the vierbein form.
The vierbein form provides an orthonormal basis in $\TsM$ and it is a dual basis of $\partial_a$ such as $\eee^a\cdot\partial_b=\E^a_\mu\E^\nu_b{dx}^\mu\cdot\partial_\nu=\delta^a_b$ and is a rank-one tensor (vector) in $\TM$.
The four-dimensional volume form is represented using vierbein forms as 
\begin{align}
\vvv&:=\frac{1}{4!}\epsilon_{\bcdots\bcdots}\hspace{.1em}\eee^\bcdot\wedge\eee^\bcdot\wedge\eee^\bcdot\wedge\eee^\bcdot
=\deteps\hspace{.2em}dx^0{\wedge}dx^1{\wedge}dx^2{\wedge}dx^3.\label{volumeF}
\end{align}
Similarly, the two-dimensional surface form is defined as
\begin{align*}
\SSS_{ab}:=\frac{1}{2}\epsilon_{ab\bcdots}\eee^\bcdot\wedge\eee^\bcdot.
\end{align*}

Connection one-form $\www$ concerning the $SO(1,3)$ group, namely the spin-connection form,  is introduced.
We define a $SO(1,3)$-covariant differential for one-form object $\aaa\Omega^1$ using the spin-connection form as
\begin{align*}
d_\www\aaa^a:=&d\aaa^a+\cG\www^a_{~\bcdot}\wedge\aaa^\bcdot,\\
\intertext{where}
\www^a_{~b}:=&
\eta_{b\bcdot}\hspace{.1em}\omega^{~a\bcdot}_\mu\hspace{.1em}{dx}^\mu=
\eta_{b\bcdot}\hspace{.1em}\omega^{~a\bcdot}_\mu\hspace{.1em}\E^\mu_\star\eee^\star=
\eta_{b\bcdot}\hspace{.1em}\omega^{~a\bcdot}_\star\hspace{.1em}\eee^\star,
\end{align*}
and $\omega^{~ab}_{\mu}$ is a component of the spin-connection form using the standard basis.
Raising and lowering indexes are done using a metric tensor.
Two-form object 
\begin{align*}
\TTT^a:=d_\www\eee^a\in V^1(\TM)\otimes\Omega^2 (\TsM)
\end{align*}
is referred to as a torsion form.
A curvature two-form is defined owing to the structure equation as
\begin{align*}
\RRR^{ab}&:=d\www^{ab}+\cG\www^a_{~\bcdot}\wedge\www^{\bcdot b}\in V^2(\TM)\otimes\Omega^2 (\TsM).
\end{align*}
The first and second Bianchi identities are
\begin{align}
d_\www\TTT^a&=d_\www(d_\www\eee^a)=\cG\hspace{.2em}
\eta_{\bcdots}\RRR^{a\bcdot}\wedge\eee^{\bcdot}
\quad\text{and}\quad
d_\www\RRR^{ab}=0.\label{Bicnchi2}
\end{align}
The curvature form is a two-form valued rank-$2$ tensor that has a coordinate representation using the standard basis in $\TsM$ as
\begin{align*}
\RRR^{ab}&=
\sum_{c<d}\Ri^{ab}_{\hspace{.7em}cd}\hspace{.1em}\eee^{c}\wedge\eee^{d}=
\frac{1}{2}\Ri^{ab}_{\hspace{.7em}\bcdots}\hspace{.1em}\eee^\bcdot\wedge\eee^\bcdot.
\end{align*}
Tensor coefficient $R^{ab}_{\hspace{.7em}cd}$ is referred to as the Riemann-curvature tensor.
Ricci-curvature tensor and scalar curvature are defined, respectively, owing to the Riemann-curvature tensor as
\begin{align*}
R^{ab}:=\Ri^{\bcdot a}_{\hspace{.7em}\bcdot \star}\eta^{b\star}~~\textrm{and }~~
R:=\Ri^{\bcdot\star}_{\hspace{.7em}\bcdot\star}.
\end{align*}

We summarise a physical dimension of form objects appearing in this section:
We assign a length dimension to coordinate vectors $dx^\mu$ and $\eee^a$: $\fdl dx^\mu\fdr=\fdl \eee^a\fdr=L$ and set a vierbein to null dimension: $\fdl \E\fdr=1$.
On the other hand, the spin-connection has a dimension $\fdl \omega\fdr=L^{-1}$; thus, the spin-connection one-form has a null dimension $\fdl \www\fdr=1$.
An external derivative operator has a null dimension such that:
\begin{align*}
df=\frac{\partial f}{dx^\bcdot}\eee^\bcdot&\implies
\fdl{d}\fdr
\fdl{f}\fdr=\left[\frac{\partial f}{\partial{x}^\bcdot}\eee^\bcdot\right]=\fdl{f}\fdr L^{-1}L=\fdl{f}\fdr
\implies\fdl{d}\fdr=1.
\end{align*}
Two-form objects have a physical dimension such that: $\fdl \SSS\fdr=L^2$, $\fdl \TTT\fdr=L$ and $\fdl \RRR\fdr=1$.
Therefore, the Riemann curvature tensor has a physical dimension of $\fdl R\fdr=L^{-2}$.
A volume form has physical dimension of $\fdl \vvv\fdr=L^4$ and a coupling constant has a null dimension.

%
%
\subsection{Amphometric space}\label{q-metric}
The current author first introduced the amphometric tensor, denoted as $\bm\eta^{~}_\theta$  in Ref.\cite{Kurihara:2025tro}, which is the one-parameter family of the complex symmetric $n$-dimensional tensor in the $n$-dimensional smooth manifold.
The amphometric tensor coincides with the Euclidean metric tensor at $\theta=0$ and the Lorentzian metric at $\theta=\pm1$.
When two functions, e.g., the solutions of a differential equation in Euclidean and Lorentzian metric spaces, are homotopical equivalence through parameter $\theta$, we can import some properties of one function in the Euclidean metric space into another function in the Lorentzian metric space.
We utilise the amphometric space to ensure the existence of a propagator in the Lorentzian metric space owing to it in the  Euclidean metric space. 
This subsection introduce the amphometric space.

The amphometric space is pair $(\Mg,\bm\eta^{~}_\theta)$, where  $\Mg$ is the four-dimensional smooth and oriented manifold and $\bm{g}^{~}_\theta$ is a complex symmetric tensor.
We introduce complex symmetric tensor $\bm{\eta}^{~}_\theta$, namely the amphometric, and its inverse in $\M$ such that:
\begin{align*}
\left[\bm\eta^{~}_\theta\right]_{ab}&:=\textup{diag}\(1,e^{i\pi\theta},e^{i\pi\theta},e^{i\pi\theta}\),\quad
\left[\bm\eta^{-1}_\theta\right]^{ab}:=\textup{diag}\(1,e^{-i\pi\theta},e^{-i\pi\theta},e^{-i\pi\theta}\),
\intertext{yielding}
\bm\eta^{~}_\theta&=
\begin{cases}
\bm\eta^{~}_\textsc{e}:=\text{diag}[1,1,1,1],&\theta=\hspace{.8em}0\\
\bm\eta^{~}_\textsc{l}:=\text{diag}[1,-1,-1,-1],&\theta=\pm1
\end{cases},
\end{align*}
where subscripts ``\textsc{e}'' and ``\textsc{l}'' stand for Euclidean and Lorentzian, respectively.
For simplicity, the Lorentzian metric tensor is also denoted with no-subscript as $\bm\eta$.
We define the determinant square root of the amperometric tensor as
\begin{align*}
\sigma_{\hspace{-.1em}\theta}^\hlf&:=\text{det}[\bm\eta^{~}_\theta]^\hlf{\hspace{.1em}=\hspace{.1em}}e^{3i\pi\theta/2}\!,
\end{align*}
where a real part of $\sigma_{\hspace{-.1em}\theta}^\hlf$ can be negative.

The amphometric tensor defined in locally flat manifold $\M^{~}_\theta$ is lifted to global manifold $\MM^{~}_\theta$ curved in general, such that: 
\begin{align*}
\left[\bm{g}^{~}_\theta\right]_{\mu\nu}:=
\left[\bm\eta^{~}_\theta\right]_\bcdots
\E^\bcdot_\mu\hspace{.1em}\E^\bcdot_\nu,
\quad\text{yielding}\quad
ds^2:=\left[\bm{\eta}^{~}_\theta\right]^{~}_{\bcdots}
\hspace{.1em}\eee^\bcdot\otimes\eee^\bcdot
=\left[\bm{g}^{~}_\theta\right]_{\mu\nu}\hspace{.1em}dx^\mu{\otimes}\hspace{.1em}dx^\nu\!.
\end{align*}
We write, hereafter, elements of the metric tensor and its inverse as $\left[\bm\eta^{~}_\theta\right]_{ab}={\eta^{~}_{\theta}}_{ab}$  and $\left[\bm\eta^{-1}_\theta\right]^{ab}={\eta^{~}_\theta}^{ab}$, respectively.
For two real tangent vectors 
\begin{align*}
\bm{u}&=(u^0,u^1,u^2,u^3),~\bm{v}=(v^0,v^1,v^2,v^3)\in{V(\TM_\theta)},
\intertext{the amphometric tensor provides an inner product of two tangent vectors as}
\<\bm{u}|\bm{v}\>^{~}_\theta&:=
{\eta^{~}_{\theta}}_\bcdots\hspace{.1em}\bm{u}^{\bcdot}\hspace{.1em}\bm{v}^{\bcdot}=
u^0v^0+e^{i\pi\theta}\(u^1v^1+u^2v^2+u^3v^3\)\in\C,
\end{align*}
which is invariant under $SO_4$ at any fixed $\theta$-value.
In Ref.\cite{Kurihara:2025tro}, we introduced function $\kappa(\theta):=e^{i\pi(\theta-1)/2}$.
This study utilises a modified one as
\begin{align*}
\hat{\kappa}(\theta)&:=e^{i\pi\theta/2}=
\begin{cases}
\hspace{.8em}1&\theta=\hspace{.8em}0\\
\pm1&\theta=\pm1
\end{cases}.
\end{align*}

When we allow all $SO_4$ group actions at $\theta\neq\pm1$, they violate the causality at $\theta=\pm1$.
We introduce a subgroup of $\Ca(2)$, namely the causal subgroup,  preserving the causality of the physical phenomenon in the Lorentzian metric space.
The restriction of group actions in the causal subgroup corresponds to choosing one axis as the ``time'' axis and prohibiting mixing the ``time'' direction with the ``spatial'' directions in Euclidean space.
\begin{definition}[Causal subgroup]\label{Defamphoortho}
Suppose transformation $\bm\Lambda^{~}_c(\theta){\hspace{.1em}\in\hspace{.1em}}\Ca(2)\subset{SO_4}$ fulfils 
\begin{align*}
&\bm\Lambda^{~}_c(\theta=\pm1):=\bm\Lambda^{~}_b\in\SO^{\uparrow}\hspace{-.1em}(1,3) {\subset}SO(1,3)
~~\text{and}~~
\bm\Lambda^{~}_c(\theta=0):=\bm\Lambda^{~}_r{\hspace{.1em}\in\hspace{.1em}}SO_{~}^{\perp}\hspace{-.1em}(3){\subset}SO(4),
\end{align*}
where $\SO^{\uparrow}\hspace{-.1em}(1,3)$ is the subgroup of $SO(1,3)$ with positive time (future) and $SO_{~}^{\perp}\hspace{-.1em}(3)$ is the three-dimensional rotation in the sub-manifold perpendicular to the time axis.
It acts on the amphovector as
\begin{align*}
&\bm\Lambda^{~}_c:End\(\Ca(2)\):\v\mapsto\v':=\bm\Lambda^{~}_c\v.
\end{align*}
We call $\bm\Lambda^{~}_c$ as the causal transformation and $\Ca(2)\subset{SO_4}$ as the causal subgroup.
\QED
\end{definition}
\noindent
E.g., a following transformation is the $\Ca(2)$ group action in the amphometric space:
%
%
\begin{align*}
\bm\Lambda^{~}_c&:=
\bm\Lambda^{~}_b(\theta p^{~}_z)\bcdot\hspace{.1em}\bm\Lambda^{~}_r(\vartheta_1,\vartheta_2,\vartheta_3),
\end{align*}
where $\bm\Lambda^{~}_b(\theta p^{~}_z)$ is the Lorentz boost along the $z$-axis and $\bm\Lambda^{~}_r(\vartheta_i)$ is rotation with three Euler-angles $\vartheta_i$.
The boost operator disappears at $\theta=0$.

In the amphometric space, we can define spinor, vector, tensor and other physical fields and construct a consistent classical field theory.
Details of the amphometric space and Yang--Mills theory extended into the amphometric space are described in Ref.\cite{Kurihara:2025tro}.
%
%
\section{Harmonic analysis}\label{Hanaly}
 In the Euclidean space, the harmonic analysis ensures solutions of the Poisson equation owing to Green's function.
On the other hand, the universe has the metric with the Lorentzian signature.
Our strategy to treat the Poisson equation with the Lorentzian metric is to extend the harmonic analysis from the Euclidean space to the amphometric space and then investigate homotopy equivalence between solutions of the Poisson equation in the Euclidean and Lorentzian metric spaces.
We discuss the existence of Green's function in the Lorentzian metric space embedded in one end of the amphometric space; the Euclidean space is bedded at another end. 

This section first provides preliminaries of harmonic analysis in the Euclidean space, then extends it to the amphometric space $(\Mg,g^{~}_\theta)$, where $\Mg$ is an $n$-dimensional smooth, compact, oriented manifold equipping the amphometric tensor.

\subsection{Preparation}
Suppose $H_1$ and $H_2$ are the Hilbert space defined in $\TMg$, and ${\cal O}:H_1{\rightarrow}H_2$ is a linear operator.
In vector space $V(\TMg)\supset H_1,H_2$, we introduce the bilinear form $\bm{g}^{~}_\theta$ using the standard basis  as $[\bm{g}^{~}_\theta]_{ab}=g^{~}_\theta(\partial_a,\partial_b)$.
When $\zeta{\in}H_1$ exists for any $\xi{\in}D({\cal O})\subset{H_1}$\footnote{We denote a domain of operator $\bullet$ as $D(\bullet)$.} fulfiling
\begin{align*}
\left.\left.g^{~}_\theta\({\cal O}(\xi),\chi\)\right|_{H_2}=g^{~}_\theta\(\xi,\zeta\)\right|_{H_1}\!,
\end{align*}
then, $\zeta$ is uniquely provided for a given $\chi{\in}H_2$.
An \textit{adjoint operator} of ${\cal O}$ is defined  as 
\begin{align*}
\widehat{{\cal O}}:H_2{\rightarrow}H_1:\xi\mapsto\widehat{\cal O}(\xi),~~\text{yielding}~~
\left.\left.g^{~}_\theta\(\xi,\widehat{{\cal O}}(\chi)\)\right|_{H_1}=g^{~}_\theta\({\cal O}(\xi),\chi\)\right|_{H_2}\!.
\end{align*}
When $H_1=H_2=H$ and ${\cal O}(\xi)=\widehat{{\cal O}}(\xi)$ for any $\xi\in{H}$, operator ${\cal O}$ is called a \textit{self-adjoint}.

Next, we introduce the adjoint-operator in the space of $p$-form objects in $\TsMg$.
An orthonormal basis of $p$-forms in $n$-dimensional cotangent space $\TsMg$ consists of $n!/p!(n-p)!$ independent $p$-form objects in total.
We represent them using a short-hand notation as 
\begin{align*}
\eee^{i_1{\cdots}i_p}:=\eee^{i_1}\wedge\cdots\wedge \eee^{i_p}\in\Omega^p(\TsMg),
~\text{where}~
i_1<i_2<\cdots<i_p,~
i_\bullet\in\{0,\cdots,n-1\},
\end{align*}Suppose $p$-form objects $\aaa,\bbb\in\Omega^p(\TsMg)$ are expressed using this basis as
\begin{align*}
\aaa=\sum_{i_1<\cdots<i_p}a_{i_1\cdots i_p}(x)\hspace{.1em}\eee^{i_1{\cdots}i_p}=:
a_{i_{1:p}}(x)\hspace{.1em}\eee^{i_{1:p}}\!,
\end{align*}
and the same manner for $\bbb$, where multi-index $i_1{\cdots}i_p$ runs over all possible combinations of $i^{~}_\bullet$. 
In the above representation, multi-indices $i_1{\cdots}i_m$ for any $n{\geq}m\in\N$ are denoted as $i_{1:m}$.
When we use this multi-index representation, we take summations in the Einstein convention over all possible combinations of $i_1<i_2<\cdots<i_m$ avoiding a factor $1/m!$.

We define an $SO_4$ invariant bilinear form for two $p$-form objects as
\begin{align*}
&{g^{~}_\theta}(\aaa,\bbb)=
\<{\aaa|\bbb}\>^{~}_\theta:=
\eta_\theta^{\hspace{.2em}i_1j_1}\cdots\eta_\theta^{\hspace{.2em}i_pj_p}\hspace{.1em}
a^{~}_{i_{1:p}}\hspace{.1em}b^{~}_{j_{1:p}}=\<{\bbb|\aaa}\>^{~}_\theta\in\C.
\end{align*}
Direct calculations show $SO_4$ invariance of the bilinear form.
A square-root of the bilinear form with itself,
\begin{align*}
\<\aaa\>^{~}_\theta:={\<\aaa|\aaa\>^{~}_\theta}^{\hlf}\in\C,
\end{align*}
is called the \textit{pseudo norm} in this study.
We introduce the $L^2$-norm, too, as
\begin{align}
\|\aaa\|_2&:=\(\sum_{i_{1:p}}\left|a_{i_{1:p}}\right|^2\)^\hlf\!\in\R.\label{L2norm1}
\end{align}
The pseudo norm and the $L^2$-norm are equivalent when $\theta=0$ (Euclidean metric); thus, we also denote them as 
\begin{align*}
\<{\aaa|\bbb}\>^{~}_0=\<{\aaa|\bbb}\>^{~}_\textsc{e}=\|\aaa\|^{~}_2=\|\aaa\|^{~}_\textsc{e}.
\end{align*}

\begin{definition}
\textup{(Closed and Pre-closed Operators\cite{yosida2012functional})}\label{closedOp}\\
Suppose $H_1$ and $H_2$ are linear normed spaces, ${\cal O}:H_1{\rightarrow}H_2$ is a linear map, and $\overline{[D({\cal O})]}=H_1$, where $\overline{[\bullet]}$ denotes closure of a space $\bullet$.
Operator ${\cal O}$ is called the \textbf{closed operator}, when ${\cal O}$ fulfils 
\begin{align*}
&\xi_n{\in}D({\cal O})\Land\xi{\in}H_1\Land\chi{\in}H_2\Land
\(\lim_{n\rightarrow\infty}{{\cal O}(\xi_n)}=\xi\){\implies}\xi{\in}D({\cal O})\Land\({\cal O}(\xi)=\chi\).
\end{align*}
When operator $\tilde{{\cal O}}:H_1{\rightarrow}H_2$ fulfils $D(\tilde{{\cal O}}){\subset}D({\cal O})$ and $\tilde{{\cal O}}(\xi)={\cal O}(\xi)$ for any $x\in{H_1}$, $\tilde{{\cal O}}$ is called the \textbf{pre-closed operator}.\QED
\end{definition}
\noindent
We state the following lemma without proof (that is in literatures, e.g., Ref.\cite{yosida2012functional}).
\begin{lemma}\label{closure}
Under the same assumptions of \textup{\textbf{Definition \ref{closedOp}}};
\begin{enumerate}
\item When $\tilde{{\cal O}}$ is a pre-closed operator, there exists closure $\overline{[\tilde{{\cal O}}]}$, that fulfils $D(\overline{[\tilde{{\cal O}}]}){\subset}D({\cal O})$ for any closure ${\cal O}$ of $\tilde{{\cal O}}$.
\item For any linear operator ${\cal O}:H_1{\rightarrow}H_2$, adjoint-operator $\hat{{\cal O}}:H_2{\rightarrow}H_1$ is a closed operator.\QED
\end{enumerate}
\end{lemma}

\vskip 5mm
\noindent
The following remark immediately follows from \textbf{Lemma \ref{closure}}:
\begin{remark}\label{rem5}
When self-dual linear operator ${\cal O}:{H}\rightarrow{H}$ exits, it has closer in $\Mg$ by means of the $L^2$-norm.
\end{remark}
\begin{proof}
Due to the assumption, operator ${\cal O}$ is linear and self-adjoint yielding $g_\theta({\cal O}(\xi),\chi)=g_\theta(\xi,{\cal O}(\chi))$ for any $\xi,\chi\in{D({\cal O})}$.
Thus, ${\cal O}$ is a pre-closed operator having closure owing to \textbf{Lemma \ref{closure}}-\textit{2}.
\end{proof}

%
%
\subsection{Hodge-dual and Laplace--Beltrami operators}
The Laplace operator, 
\begin{align*}
\Delta:=\<\bm\partial|\bm\partial\>={\eta}^{\bcdots}\partial_\bcdot\partial_\bcdot, 
\end{align*}
is the second-order differential operator acting on a scalar function.
We generalise it to the Laplace--Beltrami operator using the {co-differential} operator owing to the Hodge-dual operator.
This subsection discusses the Laplace--Beltrami operator in the amphometric space.

First, we define the Hodge-dual operator for a $p$-form object in $\TsMg$:
\begin{definition}\textup{(Hodge-dual operator)}\label{DefHodgeDual}\\
Suppose $\Mg$ is an $n$-dimensional oriented and smooth manifold equipping the amphometric tensor and $\aaa,\bbb\in\Omega^p(\TsMg)$ are $p$-form objects.
The Hodge-dual operator, denoted as $\HD$, is defined to give
\begin{align}
\aaa\wedge\HD(\bbb)&:=
\sigma_{\hspace{-.1em}\theta}^\hlf\hspace{.1em}\<\aaa|\bbb\>^{~}_\theta\hspace{.1em}\vvv^{~}_\theta.
\label{HodgeD}
\end{align}
\QED
\end{definition}\noindent
This definition is different from the standard one by the factor $\sigma_{\hspace{-.1em}\theta}^\hlf$\!.
Determinant $\sigma_{\hspace{-.1em}\theta}$ is $SO_4$ invariant and $\Mg$ is given as the quotient space as (\ref{Lmani}); thus, the definition (\ref{HodgeD}) is independent of a choice of basis.
We can write $\HD(\aaa)$ using tensor coefficients as
\begin{align}
\HD(\aaa)&:=\sigma_{\hspace{-.1em}\theta}^\hlf\hspace{.2em}
a_{i_{1:p}}\hspace{.1em}\epsilon^{i_{1:p}}_{\hspace{1.2em}i_{p+1:n}}\hspace{.1em}
\eee^{i_{p+1:n}}=\hat{\aaa},\label{HodgeDual}
\end{align}
where $\epsilon^{~}_{0\cdots n-1}$ is the Levi-Civita tensor.

Tensor coefficients of $\hat{\aaa}$ are obtained from those of $\aaa$ like
\begin{align*}
\hat{\aaa}=&\hat{a}_{i_{1:n-p}}\hspace{.1em}\eee^{i_{1:n-p}}
\quad\text{with}\quad
\hat{a}_{i_{1:n-p}}:=\sigma_{\hspace{-.1em}\theta}^\hlf\hspace{.1em}
a_{j_{1:p}}\hspace{.1em}\epsilon^{j_{1:p}}_{\hspace{1.2em}i_{1:n-p}}\!,
\end{align*}
owing to (\ref{HodgeDual}).
An $n$-dimensional volume form is provided using the Hodge-dual operator as
\begin{align*}
\HD(1)&={\sigma_{\hspace{-.1em}\theta}^\hlf}\hspace{.2em}\vvv^{~}_\theta.
\end{align*}
The Hodge-dual operator with any $\theta\in[-1,1]$ yields
\begin{align}
\HD\(\HD(\aaa)\)&\overset{\text{(\ref{HodgeDual})}}{=}{\sigma_{\hspace{-.1em}\theta}^{~}}(-1)^{p(n-p)}\aaa;\label{HHa}
\end{align}
thus, an eigenvector of  the Hodge-dual operator has eigenvalue $\pm{\sigma_{\hspace{-.1em}\theta}^{\hlf}}$ or $\pm{i}{\sigma_{\hspace{-.1em}\theta}^{\hlf}}$ since
 \begin{align*}
\HD\(\aaa\)=\lambda\aaa&{\implies}\HD\(\HD\(\aaa\)\)
=\lambda^2\aaa\overset{\text{(\ref{HHa})}}{=}{\sigma_{\hspace{-.1em}\theta}^{~}}(-1)^{p(n-p)}\aaa
{\implies}\lambda=\pm{\sigma_{\hspace{-.1em}\theta}^{\hlf}}~\textrm{or}~\pm{i}{\sigma_{\hspace{-.1em}\theta}^{\hlf}}.
\end{align*}
 
We define \textit{co-differential operator} ${\cod}$ as
\begin{align*}
{\cod_\theta}\aaa&:=(-1)^{p}\hspace{.2em}\HDi\hspace{-.2em}\(d_\theta\HD(\aaa)\)=
{\sigma_{\hspace{-.1em}\theta}^{-1}}(-1)^{n(p+1)+1}\hspace{.2em}\HD\(d_\theta\HD(\aaa)\),
\end{align*}
yielding ${\cod}\hspace{.1em}{\cod}=0$.
We note that the Hodge-dual operator depends on the metric tensor; thus, the co-differential also does.
The co-differential operator is a {weak}-adjoint of external derivative such as 
\begin{align}
\int_{\Mg}\<{d_\theta\aaa|\bbb}\>^{~}_\theta\hspace{.1em}\vvv^{~}_\theta&=
{\sigma_{\hspace{-.1em}\theta}^{-\hlf}}\int_{\Mg}d_\theta\aaa\wedge\HD(\bbb)=
{\sigma_{\hspace{-.1em}\theta}^{-\hlf}}(-1)^p\int_{\Mg}\aaa{\wedge}d_\theta\HD(\bbb),\notag\\&=
{\sigma_{\hspace{-.1em}\theta}^{-\hlf}}\int_{\Mg}\aaa{\wedge}\HD\(\(-1\)^p\HDi\hspace{-.2em}
\(d_\theta\HD\(\bbb\)\)\),\notag\\&
=\int_{\Mg}\<{\aaa|{\cod_\theta}\bbb}\>^{~}_\theta\hspace{.1em}\vvv^{~}_\theta,\label{dabadb}
\end{align}
and similarly $\<{{\cod_\theta}\aaa|\bbb}\>\simeq\<{\aaa|d_\theta\bbb}\>$\footnote{We denote a weak equivalent relation as $\aaa\simeq\bbb$, which means $\int\!\aaa=\int\!\bbb$.}.
The Storks theorem and an assumption that $\Mg$ is compact are used here.

The \textit{Laplace--Beltrami operator} is defined using co-differential as
\begin{align*}
\Delta_\theta&:=d_\theta\hspace{.1em}{\cod_\theta}+{\cod_\theta}\hspace{.1em}d_\theta,
\end{align*}
which is a linear operator on $\Omega^p(\TsMg)$ for each $p$-form object and is commutable with the Hodge-dual operator.
We note that it is weak self-adjoint like 
\begin{align*}
\<\Delta_\theta\aaa|\bbb\>^{~}_\theta=\<d_\theta\cod_\theta\aaa|\bbb\>^{~}_\theta+\<{\cod_\theta}d_\theta\aaa|\bbb\>^{~}_\theta
&\simeq\<\aaa|d_\theta{\cod_\theta}\bbb\>^{~}_\theta+\<\aaa|{\cod_\theta}d_\theta\bbb\>^{~}_\theta=\<\aaa|\Delta_\theta\bbb\>^{~}_\theta.
\end{align*}
The Laplace--Beltrami operator for a scalar function (a zero-form object) has a representation owing to the standard basis as
\begin{align}
\Delta_\theta&=
{\eta^{~}_\theta}^{\bcdots}\partial_\bcdot\partial_\bcdot,\label{L-Btri}
\end{align}
which is nothing other than the standard Laplace operator.
In reality, it provides 
\begin{align*}
\text{(\ref{L-Btri})}&=
\left\{
\begin{array}{rl}
\(\partial_0\)^2+\(\partial_1\)^2
+\(\partial_{2}\)^2+\(\partial_{3}\)^2\!,&(\theta=\hspace{.8em}0),\\
\(\partial_0\)^2-\(\partial_1\)^2-\(\partial_{2}\)^2-\(\partial_{3}\)^2\!,&(\theta=\pm1).
\end{array}
\right.
\end{align*}
in a four-dimensional.
Thus, it is equivalent to the standard Laplace operator at $\theta=\pm1$ (Lorentzian) and $\theta=0$ (Euclidean).
Hereafter, we denote objects $\bullet^{~}_\theta$ with $\theta\rightarrow0$ as $\bullet^{~}_\textsc{e}$ and with $\theta\rightarrow\pm1$ as $\bullet^{~}_\textsc{l}$.
We state the following lemma\cite{warner1983foundations} in the Euclidean space without proof:
\begin{lemma}\label{regularity}
When $\HHH$ is a weak solution of $\Delta^{~}_\textsc{e}\hspace{.1em}\HHH=\aaa\in\Omega^p(\TsM^{~}_\textsc{e})$ fulfiling $\<\Delta^{~}_\textsc{e}\HHH,\hspace{.1em}\bbb\>^{~}_\textsc{e}=\<\aaa,\bbb\>^{~}_\textsc{e}$ for any $\bbb\in\Omega^p(\TsM)$, it is a strong solution:
\begin{align}
~^\forall\bbb\in\Omega^p(\TsM^{~}_\textsc{e}),
~^\exists\HHH,\aaa\in\Omega^p(\TsM^{~}_\textsc{e}),
~\<\Delta^{~}_\textsc{e}\hspace{.1em}\HHH,\bbb\>^{~}_\textsc{e}=\<\aaa,\bbb\>^{~}_\textsc{e}\implies
\Delta^{~}_\textsc{e}\hspace{.1em}\HHH=\aaa.
\end{align}
\end{lemma}
\begin{lemma}\label{lemma66}
Suppose $\{\aaa_n\}\in\Omega^p(\TsM_\textsc{e})$ is a sequence of smooth $p$-form objects in $\TsM_\textsc{e}$ such that $\|\aaa_n\|^{~}_\textsc{e}{\leq}c$ and $\|\Delta_\textsc{e}\aaa_n\|^{~}_\textsc{e}{\leq}c$ for any $n$ and for some constant $c>0$.
Then, a subsequence of $\{\aaa_n\}$ is a Cauchy sequence in $\Omega^p(\TsM_\textsc{e})$.
\end{lemma}
When $p$-form object $\HHH\in\Omega^p(\TsMg)$ satisfying the Laplace equation $\Delta_\theta\HHH=0$, it  is called a \textit{harmonic $p$-form}.
The remark below immediately follows from \textbf{Lamma \ref{lemma66}}:
\begin{remark}\label{finiteD}
A space of harmonic $p$-forms with the Euclidean metric is a finite dimensional for each $p\in\{0,\cdots,n-1\}$.
\QED
\end{remark}
\begin{proof}:
If $\bm {H}^p_\textsc{e}$ is an infinite-dimensional space, it has an infinite orthogonal sequence.
This orthogonal sequence must contain the Cauchy sequence owing to \textbf{Lamma \ref{lemma66}}, and it contradicts the infinite dimensional of $\bm {H}^p_\textsc{e}$.
Therefore, $\bm {H}^p_\textsc{e}$ is finite-dimensional.
\end{proof}

The harmonic form yields
\begin{align*}
\Delta_\theta\bm\HHH=0&\implies
0=\<\bm\HHH|\Delta^{~}_\theta\bm\HHH\>^{~}_\theta\simeq
\<d_\theta\bm\HHH\>^{2}_\theta+\<\cod_\theta\bm\HHH\>^{2}_\theta.
\end{align*}
At $\theta=0$, 
\begin{align*}
&\<d_\textsc{e}\bm\HHH\>^{2}_\textsc{e}+\<\cod_\textsc{e}\bm\HHH\>^{2}_\textsc{e}=0
\implies\<d_\textsc{e}\bm\HHH\>^{~}_\textsc{e}=0\Land\<\cod_\textsc{e}\bm\HHH\>^{~}_\textsc{e}=0
\end{align*}
since the $L^2$-norm is positive definite for non-zero forms in the Euclidean space.
Moreover, it is non-degenerated as
\begin{align*}
&^\forall\aaa\in\Omega^p(\TsM_\textsc{e}),~\<\aaa\>^{2}_\textsc{e}=\|\aaa\|^{2}_\textsc{e}=0\implies\aaa=0.
\end{align*}
Therefore, the Laplace equation on a $p$-form object is equivalent to a set of first-order differential and co-differential equations in the Euclidean space, such as 
\begin{align}
\Delta_\textsc{e}\bm\HHH=0\iff
d_\textsc{e}\bm\HHH=\cod_\textsc{e}\bm\HHH=0.\label{deHzero}
\end{align}
For general $\theta$ values, the pseudo-norm is neither positive definite nor non-degenerated; thus, the above proof does not work.
When the harmonic $p$-form exists at a finite value of $\theta$, it also induces (\ref{deHzero}).

%
%
\begin{remark}\label{rem6}
Suppose $\HHH\in\Omega^p(\TsMg)$ exists as a function of the pseudo norm squared in $\theta\in[-1,1]$.
When $\HHH$ is the harmonic $p$-form,
we obtain
\begin{align}
\Delta_\theta\bm\HHH=0\iff
d_\theta\bm\HHH=\cod_\theta\bm\HHH=0,\label{dHzeroQ}
\end{align}
for $\theta\in[-1,1]$.\QED
\end{remark}
\begin{proof}
In the Laplace--Beltrami equation, the amphometric tensor corresponds to a transformation from the Euclidean metric space, which acts on the Euclidean standard basis $d\bm{\xi}_\textsc{e}$ in $\TsMg$ as
\begin{align*}
\hat{K}:
&d\bm{\xi}_\textsc{e}\mapsto d{\bm{\xi}}_\theta:=\hat{\bm{K}}(\theta)\cdot d\bm{\xi}_\textsc{e}
=(d\xi_\textsc{e}^0,\hat\kappa(\theta)\hspace{.1em}d\xi_\textsc{e}^1,\cdots,
\hat\kappa(\theta)\hspace{.1em}d\xi_\textsc{e}^{n-1}),
\intertext{where $\hat{\bm{K}}$ is a ($n\!\times\!n$)-matrix such as}
&\hat{\bm{K}}(\theta):=\text{diag}[1,\overbrace{\hat\kappa(\theta),\cdots,\hat\kappa(\theta)}^{n-1}]
\implies\hat{\bm{K}}(\theta)\cdot\hat{\bm{K}}(-\theta)=\bm{1}_{\!n}.
\end{align*}
The external differential acts on the $p$-form object $\HHH$ as
\begin{align*}
d_\textsc{e}\HHH\(\<\bm{\xi}_\textsc{e}\>^2_\textsc{e}\)&=
d\xi_\textsc{e}^j\wedge\frac{\partial\HHH\(\<\bm{\xi}_\textsc{e}\>^2_\textsc{e})\)}
{\partial\xi_\textsc{e}^j}=
d\xi_\textsc{e}^j\wedge\frac{\partial{\xi}_\theta^k}{\partial\xi_\textsc{e}^j}
\frac{\partial\HHH\(\left\langle\hat{\bm{K}}(-\theta)\cdot\bm{\xi}_\theta\right\rangle^2_\textsc{e}\)}
{\partial{\xi}_\theta^k}=
d\xi_\theta^k\wedge\frac{\partial\HHH\(\<\bm{\xi}_\theta\>^2_\theta\)}{\partial\xi_\theta^k},\\&=
d_\textsc{e}\HHH\(\<\bm{\xi}_\theta\>^2_\theta\),
\end{align*}
and the same manner for $\bar{d}$; thus, we obtain
\begin{align*}
d_\textsc{e}\HHH\(\<\bm{\xi}_\textsc{e}\>^2_\textsc{e}\)=
\bar{d}_\textsc{e}\HHH\(\<\bm{\xi}_\textsc{e}\>^2_\textsc{e}\)=0{\implies}
d_\theta\HHH\(\<\bm{\xi}_\theta\>^2_\theta\)=
\bar{d}_\theta\HHH\(\<\bm{\xi}_\theta\>^2_\theta\)=0.
\end{align*}
Therefore, the remark is maintained.
\end{proof}
We embed the Euclidean and Lorentzian metric space into the amphometric space.
When the harmonic form $\HHH$ exists at $\theta=0$, we vary a $\theta$ value from zero to finite $\theta$.
Then, the Laplace--Beltrami equation and its solution varied accordingly keeping harmonic.
Here, we assume the solution's existence as a function of the pseudo norm squared, though this assumption is not trivial.
We discuss the point later in this report.

%
%
\subsection{Hodge decomposition}
A space of $p$-form objects can be decomposed into three spaces including the harmonic $p$-form which a kernel of the Laplace--Beltrami equation.
This decomposition called the Hodge decomposition, is essential in ensuring the existence of solutions for the Poisson equation.

We introduce a space of harmonic $p$-forms as
\begin{align*}
\bm {H}_\theta^p:=\textrm{Ker}\Delta_\theta\big|_{\Omega^p(\TsMg)}=
\left\{
\HHH\in\Omega^p(\TsMg)\left|\Delta_\theta\HHH=0\right.\right\}.
\end{align*}
Owing to (\ref{dHzeroQ}), we can divide the $L^2$ space in $\TsMg$ into three subs spaces including $\bm {H}_\theta^p$.
We define two spaces, $d(\Omega_\textsc{e}^{p-1})$ and ${\cod}(\Omega_\textsc{e}^{p+1})$, respectively, defined as
\begin{align*}
d(\Omega_\textsc{e}^{p-1})&:=\left\{\aaa\in\Omega^p(\TsMg)\big|
\aaa=d\bbb,\bbb\in\Omega^{p-1}(\TsMg)
\right\}\!,\\
{\cod}(\Omega_\textsc{e}^{p+1})&:=\left\{\aaa\in\Omega^p(\TsMg)\big|
\aaa=\bar{d}\bbb,\bbb\in\Omega^{p+1}(\TsMg)
\right\}\!.
\end{align*}
A space of $p$-form objects can be decomposed into three spaces, namely the Hodge decomposition.

%
%
\begin{theorem}\hspace{-.4em}\textup{\textbf{1}}\hspace{.2em}\
\textup{(Hodge decomposition in the Euclidean space)}\label{remB-H}
A space completion of $\Omega^p\(\TsM_\textsc{e}\)$ concerning the norm $\|\aaa\|^{~}_\textsc{e}$, which is denoted as $L^2\(\TsM_\textsc{e}\)$, has the following orthogonal-decomposition$:$
\begin{align*}
L^2\(\TsM_\textsc{e}\)&=\bm {H}^p_\textsc{e}{\oplus}d(\Omega_\textsc{e}^{p-1}){\oplus}{\cod}(\Omega_\textsc{e}^{p+1}).
\end{align*}
\QED
\end{theorem}
\noindent
One can find a proof of this theorem in literature, e.g., Ref.\cite{warner1983foundations}.
Here, we provide an outline of proof in the Euclidean space.
\begin{proof}: The proof consists of two steps:
\paragraph{\textbf{Step 1:}}
\begin{subequations}
\textbf{Remark \ref{rem6}} immediately yields that
\begin{align}
\text{(\ref{deHzero})}{\implies}\bm {H}^p_\textsc{e}{\perp} d(\Omega_\textsc{e}^{p-1})\Land\bm {H}^p_\textsc{e}{\perp} {\cod}(\Omega_\textsc{e}^{p+1}).\label{Step1-2}
\end{align}
Moreover, we obtain that
\begin{align}
{^\forall}d\aaa{\in}d\(\Omega_\textsc{e}^{p-1}\)&\Land{^\forall}\cod\bbb{\in}\cod(\Omega_\textsc{e}^{p+1})
{\implies}\(\<d\aaa|\cod\bbb\>_\textsc{e}=\<dd\aaa|\bbb\>_\textsc{e}=0\)
{\implies}d(\Omega_\textsc{e}^{p-1}){\perp} {\cod}(\Omega_\textsc{e}^{p+1}),\label{Step11}
\end{align}
thus, three spaces, $\bm {H}^p_\textsc{e}, d(\Omega_\textsc{e}^{p-1})$ and $ {\cod}(\Omega_\textsc{e}^{p+1})$, are orthogonal to each other.
\end{subequations}

Next, we show 
\begin{align*}
\(\overline{\left[\bm {H}^p_\textsc{e}\oplus d(\Omega_\textsc{e}^{p-1}) \oplus {\cod}(\Omega_\textsc{e}^{p+1})\right]}\)^c=\varnothing,
\end{align*}
where $(\bullet)^c$ is complement concerning $L^2\(\TsM_\textsc{e}\)$. 
When $\aaa\in\Omega_\textsc{e}^p$ is orthogonal to these three spaces, we obtain that
\begin{align*}
{^\forall}\aaa \Land{^\forall}\bbb\in{d}(\Omega_\textsc{e}^{p-1})& {\implies}
\(0=\<\aaa|d\bbb\>_\textsc{e}\overset{\text{(\ref{dabadb})}}{=}
\<\cod\aaa|\bbb\>_\textsc{e}\)
{\implies}\cod\aaa=0\overset{\text{(\ref{dHzeroQ})}}{{\implies}}\aaa\in\cod(\Omega_\textsc{e}^{p+1}),\\
{^\forall}\aaa \Land{^\forall}\bbb\in\cod(\Omega_\textsc{e}^{p+1})& {\implies}
\(0=\<\aaa|\cod\bbb\>_\textsc{e}\overset{\text{(\ref{dabadb})}}{=}
\<d\aaa|\bbb\>_\textsc{e}\)
{\implies}d\aaa=0\overset{\text{(\ref{dHzeroQ})}}{{\implies}}\aaa\in\bm {H}^p_\textsc{e},
\end{align*}
and that contradicts the assumptions.
Therefore,  we obtain that 
\begin{align}
\overline{\left[\bm {H}^p_\textsc{e}\oplus d(\Omega_\textsc{e}^{p-1}) \oplus {\cod}(\Omega_\textsc{e}^{p+1})\right]}=L^2(\Omega_\textsc{e}^p).\label{step1}
\end{align}

\paragraph{\textbf{Step 2:}}\label{step2}
To complete proof, we must show the existence of closure and that  these three spaces are closed subspaces of $L^2(\Omega_\textsc{e}^p)$.
This is not trivial because differential operator $d$ is not even a continuous operator.
In reality, \textbf{Remark \ref{rem5}} owing to \textbf{Lemma \ref{closure}}-\textit{2} ensures them.
\end{proof}
\noindent
The \dR theorem and \textbf{Theorem \ref{remB-H}.1} immediately yield a homomorphism:
\begin{align*}
d(\Omega_\textsc{e}^{p-1}){\cong}H^p_\dr(\M_\textsc{e}){\cong}H_{~}^p(\M_\textsc{e}).
\end{align*}

Next, we extent \textbf{Theorem \ref{remB-H}.1} from the Euclidean metric to the amphometric. 
%
%
\setcounter{definition}{10}
\begin{theorem}\hspace{-.4em}\textup{\textbf{2}}\label{L2norm}\hspace{.2em}
\textup{(Hodge decomposition in the amphometric space)}
A closure of $\Omega^p\(\TsMg\)$ owing to the absolute norm \textup{(\ref{L2norm1})}, which is denoted as $L^2\(\TsMg\)$, has the following orthogonal-decomposition;
\begin{align*}
L^2\(\TsMg\)&=\bm {H}^p_\theta{\oplus}
d(\Omega_\theta^{p-1}){\oplus}
{\cod}(\Omega_\theta^{p+1}),
\end{align*}
in $\theta\in[\varepsilon,1]$ with $0<\varepsilon<1$.
Here, two spaces, \textup{i.e.}, $d(\Omega_\theta^{p-1})$ and ${\cod}(\Omega_\theta^{p+1})$, are orthogonal to each other, which means 
\begin{align}
^\forall\aaa{\in}d(\Omega_\theta^{p-1})\Land^\forall\bbb{\in}{\cod}(\Omega_\theta^{p+1})&\implies
\lim_{\theta\rightarrow1}{\<\aaa|\bbb\>}=0.\label{ortho}
\end{align}
\end{theorem}
\begin{proof}
In the \textit{Step 1} in \textbf{Theorem \ref{remB-H}.1}, (\ref{Step11}) and (\ref{Step1-2}) are maintained in the amphometric with $\theta\in[-1,1]$ independent of choice of the $L^2$-norm, and (\ref{step1}) is true for the completion owing to the absolute norm (\ref{L2norm1}).
\end{proof}

\subsection{Green's operator with Euclidean metric}\label{GreenEuclid}
This section introduces a Green's operator in the \textbf{Euclidean} metric space.
Owing to \textbf{Remark \ref{finiteD}}, space $\bm {H}^p_\textsc{e}$ can be spanned by a finite number of orthonormal base vectors $(\hhh^1,\cdots,\hhh^m)$.
A vector rank $m$ is provided by the $p$'th cohomology of $\bar{\bm {H}}^p_\textsc{e}$.
We introduce \textit{harmonic projection operator} $H$ using orthonormal basis $\hhh^i$  such as
\begin{align*}
H\aaa:=&\sum_{i=1}^{m}\<\aaa|\hhh^i\>_\textsc{e}\hspace{.1em}\hhh^i\!,
\quad\text{yielding}\quad
H{\bcdot}H=H.
\end{align*}
We write a space of $p$-form objects orthogonal to any $\hhh^i$ as 
\begin{align}
(\bm {H}^p_\textsc{e})^\perp&=d(\Omega_\textsc{e}^{p-1}){\oplus}{\cod}(\Omega_\textsc{e}^{p+1}), \label{decomp}
\end{align}
where \emph{orthogonal} means that
\begin{align*}
^\forall\hhh\in\bm {H}^p_\textsc{e},~^\forall\!\ooo\in(\bm {H}^p_\textsc{e})^\perp \implies \<\hhh|\ooo\>_\textsc{e}=0.
\end{align*}
A complement of the harmonic operator
\begin{align}
\overline{H}&:=\bm{1}-H,\label{GreensOP}
\end{align}
where $\bm{1}$ is the identity operator.

We consider the Poisson equation for given $\sss\in\Omega^p(\TsM_\textsc{e})$ in the Euclidean space such as
\begin{align}
\Delta_\textsc{e}\FFF=\sss,\label{Poissoneqs}
\end{align}
where inhomogeneous term $\sss$ is called the source term in this report.
Since any $p$-form objects can be decompose as (\ref{decomp}), a solution  $\FFF$ can be written
\begin{align*}
\FFF=\FFF^h+\FFF^\perp,\quad\text{where}~\bullet^h:=H\bullet,~\bullet^\perp:=\overline{H}\bullet.
\end{align*}
The harmonic projection operator commutes with the Laplace--Beltrami operator; thus, we obtain
\begin{align*}
H\Delta_\textsc{e}\FFF=\Delta_\textsc{e}\FFF^h=0
\quad\text{and}\quad
\overline{H}\Delta_\textsc{e}\FFF=\Delta_\textsc{e}\FFF^\perp=\sss^\perp.
\end{align*}
Therefore, any Poisson equation can be written as
\begin{align}
(\text{\ref{Poissoneqs}})&\implies\Delta_\textsc{e}\FFF^\perp=\sss^\perp
\end{align}
In conclusion, there exists the operator, namely Green's operator, such that:
\begin{align}
G:\Omega^p(\TsM_\textsc{e})\rightarrow&(\bm {H}^p_\textsc{e})^\perp
:\sss\mapsto\FFF^\perp:=G(\sss)=\left\{
\begin{array}{cl}
0&\(\sss\in\hspace{.5em}\bm {H}^p_\textsc{e}\hspace{1em}\),\\
(\Delta_\textsc{e})^{-1}\hspace{.1em}\sss&
\(\sss\in(\bm {H}^p_\textsc{e})^\perp\).
\end{array}
\right.\label{GreensFe}
\end{align}
The existence of the Laplace--Beltrami operator inverse is ensured in the Euclidean space.
Therefore, Green's operator provides a solution of the Poisson equation (\ref{Poissoneqs}) as
\begin{align}
\FFF&=G\(\sss^\perp\).\label{SolofPoissone} 
\end{align}
Direct calculations show that Green's operator is commutable for $d$, ${\cod}$, and $\Delta_\textsc{e}$; thus, it is a weak self-adjoint operator.

We introduce a linear bounded operator $L$ as
\begin{align*}
L:\Omega^p(\TsM_\textsc{e})\otimes\Omega^p(\TsM_\textsc{e})\rightarrow\R
   :(\FFF,\aaa)\mapsto L_\FFF(\aaa):=\<\FFF|\aaa\>_\textsc{e}.
\end{align*}
When $p$-form object $\aaa\in\Omega^p(\TsM_\textsc{e})$ fulfils relation $L_\FFF(\Delta_\textsc{e}\aaa)=\<\sss^\perp|\aaa\>_\textsc{e}$, we obtain that
\begin{align*}
L_\FFF(\Delta_\textsc{e}\aaa)-\<\sss^\perp|\aaa\>_\textsc{e}&=
\<\FFF|\Delta_\textsc{e}\aaa\>_\textsc{e}-\<\sss^\perp|\aaa\>_\textsc{e}
=\<\Delta_\textsc{e}\FFF-\sss^\perp|\aaa\>_\textsc{e}=0.
\end{align*}
Thus, we obtain another representation of the Poisson equation (\ref{Poissoneqs}) as
\begin{align*}
\Delta_\textsc{e}\FFF&=\sss^\perp
{\iff}
L_\FFF(\Delta_\textsc{e}\aaa)=\<\sss^\perp|\aaa\>_\textsc{e}.
\end{align*}
Therefore, we can represent a weak solution of the Poisson equation using operator $L$ as 
\begin{align}
L_\FFF(\aaa)=\<\FFF|\aaa\>_\textsc{e}
=\left\langle G\(\sss^\perp\)|\aaa\right\rangle_\textsc{e}.
\end{align}

\subsection{Green's operator in amphometric space}\label{GreenOpQ}
This section extends the Hodge decomposition and Green's operator into the amphometric space.
The solvability of general second-order elliptic differential equations in $\R^n$ has been proven using the method of continuity and the method of \emph{a priori} estimation, both introduced by S. Bernstein\cite{Bernstein1910}.
(See also, e.g. Refs.\cite{sauvigny2006partial,krylov2008lectures}.)

First, we treat a simple equation with a given square-integrable function $f(x,\theta)$ such that:
\begin{align}
\Delta_\theta\hspace{.1em}u(x,\theta)=f(x,\theta)\label{W22equation}
\end{align}
in the amphometric space.
We discuss the existence of the solution for the equation (\ref{W22equation}) in $\theta\in[-1,1]$.
%
%
Consider the Poisson equation given in $(\Mg,\bm{g}^{~}_\theta)$ with an $SO_n$ symmetry.
Suppose 
\begin{align*}
u_\textsc{e}(x)&:=u(x,\theta=0){\in}W^{2,2},{\quad}x\in\TM_\theta
\end{align*}
is a solution of equation \textup{(\ref{W22equation})} at $\theta=0$ for given $f_\textsc{e}:=f(x,\theta=0){\in}L^2$, where $W^{2,2}$  is the Sobolev space of twice-differentiable square-integrable functions.
Suppose functions $u({x},\theta)$ and $f({x},\theta)$ are given through the pseud norm of position one-form $\xxx\in\Omega^1(\TsMg)$ such that:
\begin{align}
x:=\<\xxx,\xxx\>_{\!\theta}^{\hspace{.3em}\hlf}.
\end{align}
The Poisson equation (\ref{W22equation}) at $\theta\!=\!0$ is an elliptic-type second order differential equation; thus, solution $u$ of the equation fulfils the \textit{a priori} estimation such that:
\begin{align}
\|u_\textsc{e}\|^{~}_\textsc{e}{\leq}N\hspace{.1em}\|\Delta_\textsc{e}\hspace{.1em}u_\textsc{e}\|^{~}_\textsc{e}
=N\hspace{.1em}\|f_\textsc{e}\|^{~}_\textsc{e},\label{apriori}
\end{align}
where $N{\in}(0,\infty)$. 
Then, we consider to extent this solution to the finite values of $\theta$.
We prepare a sequence of $\theta\in[0,1]$ as
\begin{align}
\theta_i=\{\theta_0,\cdots,\theta_m\}\quad\text{with}~0=\theta_0<\theta_1<\cdots<\theta_m\leq1,~\delta_i:=\theta_{i+1}-\theta_{i}\ll1.
\end{align}
Owing to the assumption, $u(x,\theta)$ and $f(x,\theta)$ are smooth functions concerning the parameter $\theta$. 
Thus, we have the Taylor expansion on a standard basis.
When (\ref{apriori}) is still true at $\theta=\theta_1$, we can manipulate the equation as
\begin{align*}
\Delta_{\theta_1}\hspace{.1em}u(x,\theta_1)&=
\(\Delta_{\theta_0}+i{\pi}e^{i\pi\theta_0}\delta_0\sum_{j=1}^{n-1}\(\partial_j\)^2\)
\(u(x,\theta_0)
+\delta_0\left[\partial_\theta u(x,\theta)\right]_{\theta\rightarrow\theta_0}\)\!,\\
&\simeq f(x,\theta_0)+\delta_0\left[\(
\Delta_{\theta}\partial_\theta+i\pi
e^{i\pi\theta}\partial_\theta\sum_{j=1}^{n-1}\(\partial_j\)^2
\) u(x,\theta)\right]_{\theta\rightarrow\theta_0}\!,\\
&=f(x,\theta_0)+\delta_0\left[\partial_\theta\(
\Delta_{\theta}u(x,\theta)\)\right]_{\theta\rightarrow\theta_0}\!,\\
&=f(x,\theta_0)+\delta_0\left[
\partial_\theta f(x,\theta)\right]_{\theta\rightarrow\theta_0}\!,\\
&\simeq f(x,\theta_1).
\end{align*}
Here, $A\simeq B$ means $\mathcal{O}(|A-B|)=\mathcal{O}(\delta_j^2)$. 
Thus, function $u(x,\theta_1)$ is the solution of the equation (\ref{W22equation}).
We can continue this manipulation for $\theta_i$ until $u(x,\theta_i)$ fulfils the \textit{a priori} estimation (\ref{apriori}) at $\theta=\theta_i$.
Solution $u(x,\theta)$ can be extended until $\theta=\theta_i$. 

Next, we consider the Green's function in the amphometric space.
We set this orthogonal basis in the amphometric space as follows: 
\begin{enumerate}
\item At $\theta=0$, base vectors in $\Omega^p(\TsMg)$ coincide with those in the Euclidean space. 
\item  The term \emph{orthogonal} is used by means of (\ref{ortho}).
\item Completion is taken owing to the absolute norm $\|\aaa\|^{~}_2$ for $p$-form object $\aaa\in\Omega^p_{~}(\TsMg)$.
\end{enumerate}
An $L^2$ space is extended to the amphometric space and denoted as $L^2(\Omega_\theta^p)=\bm{H}_\theta^p\oplus(\bm {H}^p_\theta)^\perp$.

For the Green's operator in the amphometric space, we follow discussions for the Euclidean space in \textbf{section \ref{GreenEuclid})}.
The Poisson equation in the amphometric space is provided as
\begin{align}
\Delta_\theta\FFF&=\sss^\perp\tag*{(\ref{Poissoneqs})$'$}.\label{Poissoneqthe}
\end{align}
We similarly introduce the Green's function in the amphometric space as
\begin{align}
G_\theta:\Omega^p(\TsMg)\rightarrow(\bm {H}^p_\theta)^\perp:\sss\mapsto&\FFF^\perp:=G_\theta(\sss)\notag=\left\{
\begin{array}{cl}
0&\(\sss\in\hspace{.5em}\bm {H}^p_\theta\hspace{1.1em}\)\\
{\Delta_\theta}^{-1}\hspace{.1em}\sss&
\(\sss\in(\bm {H}^p_\theta)^\perp\)
\end{array}
\right.\tag*{(\ref{GreensFe})$'$}\!,\label{GreensFthe}
\intertext{yielding}
\FFF&=G_\theta(\sss^\perp).
\tag*{(\ref{SolofPoissone})$'$}\label{SolofPoissonthe} 
\end{align}
When the Laplace--Beltrami operator inverse exists (up to $\theta_j$, which fulfils (\ref{apriori})), we can extend the  Euclidean Green's function into the amphometric space. 

%
%
\begin{example}(Simple solutions of the Poisson equation I)\label{exPoi}\\
A scalar function in $n$-dimension ($n\geq2$) of a type,
\begin{align}
u_\alpha(\xi,\theta)&:=\({\eta^{~}_{\theta}}_\bcdots\hspace{.1em}\xi^\bcdot\xi^\bcdot\)^\alpha=
\(e^{i\pi\theta}\(\xi^0\)^2-\cdots-\(\xi^{n-1}\)^2\)^\alpha\!,\label{uAlpha}
\end{align}
is considered, where $\bm{\xi}:=(\xi^0,\cdots,\xi^{n-1})$ is a coordinate vector in $\TMg$.   
Direct calculations show the function $u_\alpha(x,\theta)$ is a solution of the Poisson equation of 
\begin{align*}
\Delta_\theta u_\alpha(\xi,\theta)=&f_\alpha(\xi,\theta),
\intertext{where}
f_\alpha(\xi,\theta):=&-2\hspace{.1em}\alpha\(n+2\(\alpha-1\)\)
\({\eta^{~}_{\theta}}_{\bcdots}\hspace{.1em}\xi^\bcdot\xi^\bcdot\)^{\alpha-1}.
\end{align*}
When $\alpha=0$, $u_0(\xi,\theta)$ is a constant function and the equation is trivial.
When $\alpha=-1$, we obtain the Laplace equation with $f_{-1}(\xi,\theta)=0$.
The Laplace--Beltrami operator is a linear operator and the Poisson equation is an inhomogeneous equation concerning the Laplace--Beltrami operator; thus, the general solution $u(\xi,\theta)$ is given like
\begin{align}
u(\xi,\theta)&=u_{-1}(\xi,\theta)
+\(\textrm{a particular solution}\).\label{PoissonGsol}
\end{align}
For $\alpha<0$ and $\theta=0$, $u_\alpha(\xi,\theta)$ has a pole for the light-like vector $\bm{\xi}_0$, which is the vector $\bm{\xi}_0\neq0$ and $g_\theta(\bm{\xi}_0,\bm{\xi}_0)=0$).
On the other hand, $u_\alpha(\xi,\theta)$ has no pole in $\theta\in(-1,1)$.

There is another pair of the Poisson equation such that:
\begin{align*}
u(\xi,\theta)&=\log{
\({\eta^{~}_{\theta}}_{\bcdots}\hspace{.1em}\xi^\bcdot\xi^\bcdot\)}
\intertext{and}
f(\xi,\theta)&=-2(n-2)
\({\eta^{~}_{\theta}}_{\bcdots}\hspace{.1em}\xi^\bcdot\xi^\bcdot\)^{-1}.
\end{align*}
\QED
\end{example}
%
%
\begin{example}(Simple solutions of the Poisson equation I\hspace{-.1em}I)\label{SpheHarmo}\\
The regular and irregular solid harmonics are solutions to the Laplace--Beltrami equation in three-dimensional polar coordinate.
They are expressed using spherical harmonics as
\begin{subequations}\\
(Regular solid harmonics):
\begin{align}
R(\vec{r};l,m)&:=\sqrt{\frac{4\pi}{2l+1}}\hspace{.1em}r^l\hspace{.1em}Y(\vartheta,\varphi;l,m),\label{RSH}
\intertext{(Irregular solid harmonics):}
I(\vec{r};l,m)&:=\sqrt{\frac{4\pi}{2l+1}}\hspace{.1em}r^{-l-1}\hspace{.1em}Y(\vartheta,\varphi;l,m),\label{ISH}
\end{align}
where $Y(\vartheta,\varphi;l,m)$ is spherical harmonics and $\vec{r}:=(r\sin{\vartheta}\cos{\varphi},r\sin{\vartheta}\sin{\varphi},r\cos{\vartheta})$ is a position vector in three-dimension.
 \end{subequations}
We use $\vartheta$ for a polar-angle to distinguish the parameter $\theta$ of the amphometric.
The solid harmonics is also a solution to the four-dimensional Laplace--Beltrami equation since it does not have $t:=x^0$ dependence; thus, the solid harmonics constructs a general solution of the Poisson equation in four-dimensions instead of $u_{-1}(\xi,\theta)$ in (\ref{PoissonGsol}).
\QED
\end{example}

%
%
\section{Fourier transformation and Green's function}\label{Ftrans}
This section gives fundamental solutions to the Poisson equation in the amphometric space utilizing the Fourier transformation and Green's function.
The Fourier transformation requires two versions of inertial manifold\cite{talagrand_2022}.
We introduce an adjoint manifold, namely a \emph{momentum manifold}, associating the inertial manifold.
The inertial manifold is called the \emph{configuration manifold} in this context.
The Fourier transformation is a map between $L^2$ functions in tangent and cotangent bundles for the configuration and momentum manifolds.
Green's function introduced in the preceding section is an inverse of the Laplace--Beltrami operator and provided using the Fourier transformation with the Dirac $\delta$-function (distribution).

\subsection{Fourier transformation in amphometric space}
We introduce a momentum manifold concerning the inertial bundle:
 In this context, the inertial manifold itself, the $n$-dimensional smooth manifold $(\Mg,\bm{\eta}^{~}_\theta)$ with the amphometric, is called as the \emph{configuration manifold}.
The \emph{momentum manifold} $(\tMg,\bm{\eta}^{~}_\theta)$ is another version of the $\theta$-extended Riemannian manifold  sharing the same amphometric $\bm\eta^{~}_\theta$ with the configuration manifold.
Position vectors, vectors pointing to a point in $\TMg$ and $\TsMg$, are introduced as
\begin{align*}
\left\{
\begin{array}{clll}
\p&:=p^\bcdot{\partial}/{\partial\xi^\bcdot}&=p^\bcdot\partial_\bcdot&\in\TMg,\\
\rrr&:=r_\bcdot\hspace{.1em}d\xi^\bcdot&=r_\bcdot\eee^\bcdot&\in\TsMg.
\end{array}
\right.
\end{align*}
Similarly, position vectors in $\TtMg$ and $\TstMg$ are, respectively,
\begin{align*}
\left\{
\begin{array}{clll}
\tbr&:=\tr^\bcdot{\partial}/{\partial\tilde{\xi}^\bcdot}&=\tr^\bcdot\tilde{\partial}_\bcdot&\in\TtMg\\
\tppp&:=\tp_\bcdot d\tx^\bcdot&=\tp_\bcdot{\teee^\bcdot}&\in\TstMg,
\end{array}
\right.
\end{align*}
where $\tilde{\partial}_\bullet:={\partial}/{\partial\tx^\bullet}$ and $\teee^\bullet$ are, respectively, orthonormal bases in $\TtMg$ and $\TstMg$.
A position vector in the momentum manifold is called a momentum vector.
Base vectors $\teee^a$ and $\tilde{\partial}_b$ are dual bases to each other such that $\teee^a\cdot\tilde{\partial}_b=\delta^a_b$.
We assign physical dimension to vector components as $[p]=[\tp]=E$ and $[r]=[\tr]=L$; thus, position vectors have physical dimension of $[\p]=E/L$, $[\rrr]=L^2$, $[\tbr]=L/E$, and $[\tppp]=E^2$, respectively. 

\begin{definition}\textup{(Fourier transformation)}\label{pFT}\\
We introduce a $SO_{\!n}$-invariant bilinear form between a tangent vector and a one-form onject in the configuration space in the configuration space as 
\begin{align*}
\bm{a}&=a^\bcdot\partial_\bcdot\in\TMg,~\bbb=b_\bcdot\eee^\bcdot\in\TsMg\rightarrow
\<\bm{a}|\bbb\>:=a^\bcdot b_\bcdot=\<\bbb|\bm{a}\>\in\C,
\intertext{and similally in the momentum space as}
\tilde{\bm{a}}&=\tilde{a}^\bcdot\partial_\bcdot\in\TtMg,~\tilde\bbb=\tilde{b}_\bcdot\eee^\bcdot\in\TstMg\rightarrow
\<\tilde{\bm{a}}|\tilde\bbb\>:=\tilde{a}^\bcdot\tilde{b}_\bcdot=\<\tilde\bbb|\tilde{\bm{a}}\>\in\C.
\end{align*}
We define a Fourier transformation as a map between $L^2$-functions in the configuration manifold such that: 
\begin{subequations}
\begin{align}
&\pF:L^2(\TsMg){\rightarrow}L^2(\TMg):f(\br){\mapsto}\pF(f)(\bp):=
\int_\TsMg\exp{\(+\frac{i}{\hbar}\<\rrr|\bp\>\)}f(\br)\vvv,\label{FT}
\intertext{and in the momentum manifold:}
&\widetilde{\pF}:L^2(\TstMg){\rightarrow}L^2(\TtMg):
f(\btp ){\mapsto}\widetilde{\pF}(f)(\btr):=
\frac{1}{\(2\pi\hbar\)^{n}}\int_\TstMg
\exp{\(-\frac{i}{\hbar}\<\tbr|\tppp\>\)}
f(\btp)\tvvv,\label{FTi}
\intertext{where}  
&\vvv:=\frac{1}{n!}\epsilon_{\bcdot\cdots\bcdot}\hspace{.1em}
\eee^\bcdot\wedge\cdots\wedge\eee^\bcdot\notag~~\text{and}~~
\tvvv:=\frac{1}{n!}\epsilon_{\bcdot\cdots\bcdot}\hspace{.1em}
\teee^\bcdot\wedge\cdots\wedge\teee^\bcdot\notag
\end{align}
are volume-forms in configuretion and momentum spaces, respectively.
\end{subequations}\QED
\end{definition}
\noindent
The Fourier transformation in the momentum space is defined as the inverse Fourier transformation.
The tangent-cotangent vector bilinear form does not include the metric tensor; thus, it is independent of the $\theta$ value of the amphometric.

An argument of the exponential function has a null dimension.
Physical dimension of the Fourier transformation is 
\begin{align*}
\left[\frac{\pF(f)(\bp)}{f(\br)}\right]=
\left[\frac{f(\btp)}{\widetilde{\pF}(f)(\btr)}\right]=L^n.
\end{align*}


A Fourier transformation (\ref{FT})  has the expression in $\TsMM$ (generally curved) manifold such that:
\begin{align*}
\pF(f)(p)
=&
\int_\TsMM{e^{i\hspace{.1em}{r^{~}_\bcdot}{p^\bcdot}/\hbar}}
\hspace{.2em}f^{~}_{\E}(x)dx^{0}{\wedge}\cdots{\wedge}dx^{n-1},
\intertext{where}
f^{~}_{\E}(x\in\TsMM):=&
\frac{1}{n!}\epsilon^{\mu^{~}_\textsc{e}\cdots\mu^{~}_{n}}
\E_{\mu^{~}_\textsc{e}}^0(x)\cdots\E_{\mu^{~}_{n}}^{n-1}(x)f(x)=\mathrm{det}[\E]
f(x),
\end{align*}
and it is $GL(n)$ invariant.

\begin{definition}\textup{($\CMm$ transformation)}\label{CM}\\
A transformation of vectors from the configuration space to the momentum space (and vice versa), namely the $\CMm$-transformation, is defined as
\begin{align*}
\CMg&:V^1(\TMg){\rightarrow}\Omega^1(\TstMg):\p\mapsto{\CMg}\(\p\):=\tppp,~\<\bp|\tppp\>:=\<\bp|\bp\>^{~}_\theta,\\
\MCg&:V^1(\TtMg){\rightarrow}\Omega^1(\TsMg):
\tilde{\br}\mapsto{\MCg}\(\tr\):=\btr,~\<\br|\btr\>:=\<\br|\br\>^{~}_\theta.
\end{align*}\QED
\end{definition}
\noindent
For a given $\p\in V^1(\TM)$, one-form object $\tppp\in\Omega^1(\TstM)$ fulfiling $\<\bp|\tppp\>:=\<\bp|\bp\>^{~}_\theta={\eta^{~}_\theta}_{\bcdots}p^\bcdot p^\bcdot$ exists.
The transformation has a component representations using the standard basis as
\begin{align}
\tppp&:=\eta^{~}_{\theta\bcdots}\hspace{.1em}{p^\bcdot}\eee^\bcdot~\textrm{and}~~
\btr:=\eta^{\hspace{.2em}\bcdots}_{\theta}\hspace{.1em}{\tilde{r}_\bcdot}\hspace{.1em}\partial_\bcdot.\label{tptr}
\end{align}
\noindent
Coordinate components of position vectors $\tp_a$ and $\tr^a$ with $a\in\{0,\cdots,n-1\}$ are given from $p^a$ and $r_a$ using the amphometric tensor as (\ref{tptr}); thus, they are complex-valued vectors.
We assume that manifold $\tMg$ is a complex manifold covered by open sets homeomorphic to $\C^n$, and two local coordinates are related to each other owing to the holomorphic function.

\begin{figure*}[t] 
\centering
\includegraphics[width={13cm}]{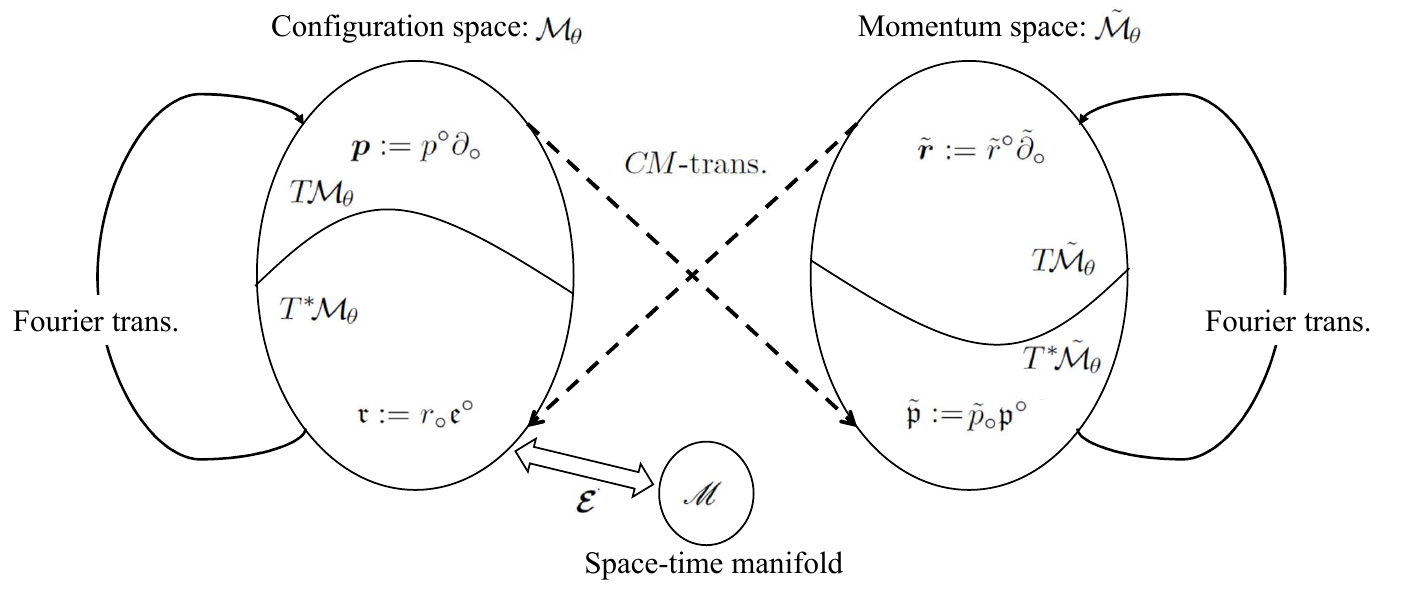}%
 \caption{
A schematic view of configuration and momentum spaces and transformations among their tangent and co-tangent spaces.
}\label{Fourier}
\end{figure*}

\begin{definition}\textup{($C\hspace{-.2em}M$-Fourier(--Laplace) transformation)}\label{FLT}\\
The $C\hspace{-.2em}M$-Fourier transformation and its adjoint transformation are defined as a sequential transformation of Fourier and $\CMm $transformations such that:
\begin{subequations}
\begin{align}
\FTg\hspace{1.em}:=\CMg\bcdot\pF:f(\br){\mapsto}\hspace{1.em}\FTg(f)(\btp ):=&
\left|\deteta\right|\int{e^
{+i{\eta^{~}_\theta}^{\hspace{-.1em}\bcdots}r_\bcdot\tp_\bcdot/\hbar}}f(\br)\vvv,
=:\tilde{f}(\btp ),\label{FL1}\\
=&\left|\deteta\right|\int_\TsMMg{e^
{+i{\eta^{~}_\theta}^{\hspace{-.1em}\bcdots}r_\bcdot\tp_\bcdot/\hbar}}f^{~}_{\E}(\bm{x})
dx^0{\wedge}\cdots{\wedge}dx^{n-1},
\notag\\
\FTg^{-1}:=\MCg\bcdot\widetilde{\pF}:
\tilde{f}(\btp ){\mapsto}\FTg^{-1}(\tilde{f})(\br):=&
\frac{\left|\deteta\right|}{\(2\pi\hbar\)^{n}}\int_\TstMg{e^
{-i{\eta^{~}_\theta}^{\hspace{-.1em}\bcdots}r_\bcdot\tp_\bcdot/\hbar}}\tilde{f}(\btp )\tvvv.
\label{FL2}
\end{align}
\end{subequations}
The Jacobian matrix $\left|\deteta\right|$ is unity for any $\theta$; thus, we omit it in formulae, hereafter.\QED
\end{definition}
\noindent
In summary, we define the $C\hspace{-.2em}M$-Fourier transformation using a two-step approach: the first step is a Fourier transformation from the cotangent bundle to the tangent bundle within configuration or momentum manifold, then the second step is the $C\hspace{-.2em}M$-transformation between two manifolds, as shown schematically in Figure \ref{Fourier}.
The configuration manifold is equivalent to the inertial manifold. 

Suppose $\psi(x)\in\Omega^0(\TsMg)$ is a static wave-function of non-relativistic quantum mechanics and $\hat{\bp}:=\hat{p}^\bcdot\partial_\bcdot\in{V}^1(\TMg)$ is a unit vector such that $\left|{\eta^{~}_{\theta}}_{\bcdots}\hspace{.1em}\hat{p}^\bcdot\hat{p}^\bcdot\right|=1$ and $\hat{p}^0=0$.
We note that parameter $\theta$ does not appear in this system.
Contraction of the wave function concerning the unit vector $\hat{\bp}$ denoted as  $\iota^{~}_{\hat{\bp}}$ provides
\begin{align*}
\iota^{~}_{\hat{\bp}}\psi(\bxi):=\hat{p}^\bcdot\(\frac{\partial\psi(\bxi)}{\partial{\xi}^\bcdot}\);
\end{align*}
thus, $\partial_\bcdot\psi$ yields a three-dimensional momentum vector.
In the configuration space, a momentum operator is a differential operator, and a position operator is a multiplying operator, and vice versa in the momentum space.
Our formulation is compatible with the standard quantum mechanical interpretation.

At $\theta=0$ and $\theta=\pm1$, an integral kernel of transformations (\ref{FL1}) and (\ref{FL2}) has a pure imaginary exponent; thus, they are Fourier transformation and its inverse.
In $\theta\in(-1,1)\setminus\{0\}$, an integral kernel has a real exponent, and they are the Fourier--Laplace transformation, which is an analytic continuation of the Fourier transformation.
The existence of the inverse transformation is ensured after the appropriate deformation of an integration line.
The adjoint transformation (\ref{FL2}) is an inverse Fourier transformation at $\theta=0$ and $\theta=\pm1$; thus, we obtain that
\begin{align*}
\(\FT_{0}(\tilde{f})(\br)\)^{-1}&=\(\FT_{\hspace{-.2em}\pm\hspace{-.1em}1}(\tilde{f})(\br)\)^{-1}=f^{~}_\E(\bx).
\end{align*}

\subsection{Delta distribution and propagator}
\begin{subequations}
We introduced Green's operator to describe a solution of Poisson equation \ref{SolofPoissonthe} in \textbf{Section \ref{GreenOpQ}}.
The Laplace-Beltrami operator inverse is necessary to obtain Green's operator.
Green's operator acting on the Dirac $\delta$-distribution as the source term is called Green's function (or \textit{propagator}) and provides the Laplace-Beltrami operator inverse; thus, it gives the solution to the given Poisson equation.

We define the Dirac $\delta$-distribution in a one-parameter space as
\begin{align}
\delta(x):=&
\frac{1}{2\pi}\int_{-\infty}^{\infty}dk\hspace{.1em}e^{ikx},\label{delta1}
\intertext{and in an $n$-dimensional vector space as}
\delta^n(\bx):=&\prod^{n-1}_{j=0}\delta(x_j),\label{deltan}
\end{align}
where $\bx=\{x_0,\cdots,x_{n-1}\}$.
\end{subequations}
Following remarks immediately follow from the definition:
\begin{remark}\textup{(Properties of the Dirac $\delta$-distribution)}\label{deltaRM}\\
Suppose $a,b,x\in\R$ and $z\in\C$.
\begin{enumerate}
\item For any $\varphi(z){\in}C^0(\C)$, we have 
\begin{align*}
&0\notin[a,b]\implies\int_a^b\varphi(x)\delta(x)dx=0,\\
&0\in[a,b]\implies\int_a^b\varphi(x)\delta(x)dx=
\frac{1}{2\pi i}\oint\frac{\varphi(z)}{z}dz=\varphi(0),
\end{align*}
where an integration contour is a clock wise circle around $z=0$.
\item A derivation of the Dirac $\delta$-distribution yields a recurrence formula such that:
\begin{align*}
x\frac{d^j\delta}{dx^j}(x)=-j\frac{d^{j-1}\delta}{dx^{j-1}}(x),
~~\textrm{where}~~1{\leq}j\in\Z.
\end{align*}
\item Suppose $\bm{F}(\bx)=\(F^0(\bx),\cdots,F^{n-1}(\bx)\)\in\R^n$ is a holomorphic function with 
\begin{align*}
\bm{F}(\bm{0})=\bm{0}\quad\text{and}\quad
\textup{det}[\partial{\bm{F}(\bx)}/\partial{\bx}|_{\bx=\bm{0}}]\neq\bm{0}. 
\end{align*}
Then, we have
\begin{align*}
\delta\(\bm{F}(\bx)\)dx^0{\wedge}\cdots{\wedge}dx^{n-1}&=
\textup{det}\left[\frac{\partial{\bm{F}(\bx)}}{\partial\bx}\right]^{-1}
\delta\(\bx\)
dx^0{\wedge}\cdots{\wedge}dx^{n-1}.
\end{align*}
Thus, functions in the amphometric space yields
\begin{align*}
\delta(\bm{\eta}\bx){\sigmath}dx^0{\wedge}\cdots{\wedge}dx^{n-1}=
\delta\(\bx\){\deteps}dx^0{\wedge}\cdots{\wedge}dx^{n-1}.
\end{align*}
\end{enumerate}
\end{remark}
\noindent

We consider the Poisson equation \ref{Poissoneqthe} for a $p$-form object in the $n$-dimensional amphometric space.
We denote the orthonormal basis in $L_2(\Omega_\theta^p)=\bm{H}_\theta^p\oplus(\bm {H}^p_\theta)^\perp$ as
\begin{align}
\(\vvvp^1,\cdots,\vvvp^N\)
=\(\hhh^1,\cdots,\hhh^m,\bar{\hhh}^1,\cdots,\bar{\hhh}^{N-m}\),\label{orthnor}
\end{align}
where $\{{\hhh}^j\}$ and $\{\bar{\hhh}^j\}$ are orthonormal bases of $\bm {H}^p_\theta$ and $(\bm {H}^p_\theta)^\perp$, respectively, and$N$ is a rank of $\Omega^p(\TsM_\theta)$.
The Poisson equation has a representation using the orthonormal basis (\ref{orthnor}) such that:
\begin{align*}
\Delta_\theta(\bx)\FFF&=\sss^\perp\implies
(-1)^{n(p+1)+1}{\eta^{~}_\theta}^{\bcdots}\(\partial_\bcdot\partial_\bcdot
F_\star\(\bx\)\)\bar{\hhh}^\star=s_\bcdot\(\bx\)\bar{\hhh}^\bcdot,
\intertext{where}
\FFF:=&F_\bcdot\(\bx\)\vvvp^\bcdot=\sum_{a=1}^{m}F_a\(\bx\)\vvvp^a
~\text{and}~
\sss^\perp\hspace{-.1em}\(\bx\):=s_\bcdot\(\bx\)\bar{\hhh}^\bcdot=
\sum_{a=1}^{N-m}s_a\(\bx\)\bar{\hhh}^a.
\end{align*}
We note that $\Delta_\theta\vvvp^a=\Delta_\theta\hhh^a=0$ for any $a\in\{1,\cdots,m\}$.
A propagator is the Green's operator acting on the Dirac $\delta$-distribution.
Due to \ref{GreensFthe}, it is given as
\begin{align*}
\Delta_\theta \GGG_\theta=\delta(\bm{x})\bar{\hhh}\implies
{\GGG_\theta}=G_\bcdot\(\delta(\bx)\)\vvvp^\bcdot:=
\left\{
\begin{array}{cc}
0&\(1{\leq}a{\leq}m\),\\
\Delta_\theta^{-1}(\bx)\delta(\bx)\bar{\hhh}^{a-m}&
\(m+1{\leq}a{\leq}N\),
\end{array}
\right.
\end{align*}
where $\bx:=\{x^1,\cdots,x^{N-m}\}$ and $a\in\{1,\cdots,N\}$.
A formal solution of equation \ref{Poissoneqthe} is provided as
\begin{align*}
&\FFF\(\bx\)=\FFF^h\(\bm{x}\)+
\(\int_{D(\bm{y})}\hspace{-.2em}d\bm{y}\hspace{.2em}
{G_\theta}\(\delta(\bx-\bm{y})\vvvp^\bcdot\)
s_\bcdot\(\bm{y}\)\)\bar{\hhh}^\bcdot,
\intertext{yielding}
&\Delta_\theta\(\bx\)\FFF\(\bx\)=\int_{D(\bm{y})}\hspace{-.2em}d\bm{y}\hspace{.2em}
\delta(\bx-\bm{y})\sss^\perp\!\(\bm{y}\)=\sss^\perp\!\(\bx\),
\end{align*}
where $d\bm{y}:=dy^1{\wedge}\cdots{\wedge}dy^{N-m}$ and $D(\bm{y})$ is a domain of $\bm{y}$ with $D(\bx){\subseteq}D(\bm{y})$.

\begin{example}(The fundamental solution of the Klein--Gordon equation)\label{KGeq}\\
We consider the Klein--Gordon equation in the four-dimensional amphometric space with $\theta\in(-1,1)$.
For scalar function $\phi(\br)\in\Omega^0\(\TsMg\)$, the Klein--Gordon equation is provided as
\begin{align}
\(\Delta_\theta\(\br\)+\(\mu/\hbar\)^2\)\phi(\br)=&-\frac{1}{\hbar^2}f\(\br\),\label{KGeq}
\end{align}
where $0<\mu\in\R$ is a real constant, namely a particle mass in physics, and $f(\br)$ is a real-valued function.
Factor $\hbar^{-2}$ at the right-hand side of (\ref{KGeq}) is just a convention. 
In physics, operator $\Delta_\theta(r)$ has a physics dimension $\fdl\Delta_\theta(r)\fdr=L^{-2}$, and a particle mass has the energy dimension $\fdl\mu\fdr=E$.
Though the Planck constant adjusts physical dimensions in formulae, e.g., $\fdl\Delta_\theta(r)\fdr=\fdl\mu^2/\hbar^2\fdr=L^{-2}$, the theory is still classical.
The scalar function $\phi(r)$ has the energy dimension; thus, $f(\br)$ has $E^3$ dimension, keeping both sides of (\ref{KGeq}) to be the same physical dimension of $EL^{-2}$.

Suppose function $\phi(\br)$ is given as the $\CMg$-Fourier--Laplace transformation of $\tilde\phi(\btp)\in\Omega^0(\TstMg)$ as
\begin{align*}
\phi(\br):=&\CMg\left[\frac{1}{\(2\pi\hbar\)^{4}}
\int_\TstMg{e^{-i{\eta^{~}_\theta}^{\hspace{-.1em}\bcdots}
r_\bcdot\tp_\bcdot/\hbar}}\tilde{\phi}(\btp)\tvvv\right]\in\Omega^0\(\TsMg\),
\intertext{yielding}
\(\hbar^2\Delta_\theta(\br)+\mu^2\)\phi(\br)
&=\frac{1}{\(2\pi\hbar\)^{4}}
\CMg\left[
\int_\TstMg\(-{\eta^{~}_\theta}^{\hspace{-.1em}\bcdots}{\tp_\bcdot}{\tp_\bcdot}+\mu^2\)
{e^{-i{\eta^{~}_\theta}^{\hspace{-.1em}\bcdots}
r_\bcdot\tp_\bcdot/\hbar}}\hspace{.2em}
\tilde{\phi}(\btp)\tvvv\right]\!.
\end{align*}
We obtain Green's function in the configuration space by means of the $\CMg$-Fourier transformation of Green's function in the momentum space such that:
\begin{align}
{G_\theta}\(\delta^4(\br)\)&=\CMg\left[\frac{1}{(2\pi\hbar)^4}\int_\TstMg
e^{-i{\eta^{~}_\theta}^{\hspace{-.1em}\bcdots}r_\bcdot\tp_\bcdot/\hbar}\hspace{.1em}
{G_\theta}(\btp)\tvvv\right]\!,\label{KGprop}
\end{align}
where Green's function in the momentum space is
\begin{align*}
{G_\theta}(\btp)&:=
\({-{\eta^{~}_\theta}^{\hspace{-.1em}\bcdots}{\tp_\bcdot}{\tp_\bcdot}+\mu^2}\)^{-1}\in\Omega^0(\TstMg).
\end{align*}
We note that integrations (\ref{KGprop}) is well-defined since ${\eta^{~}_\theta}^{\hspace{-.1em}\bcdots}{\tp_\bcdot}{\tp_\bcdot}-\mu^2\neq0$ in $\theta\in(-1,1)$.
In reality, we obtain from (\ref{KGprop}) with the definition (\ref{delta1}) that
\begin{align*}
(\ref{KGprop})\implies\(\hbar^2\Delta_\theta\(\br\)+\mu^2\){G_\theta}\(\delta^4(\br)\)&=\CMg\left[
\frac{1}{(2\pi\hbar)^4}\int_\TstMg
{e^{-i{\eta^{~}_\theta}^{\hspace{-.1em}\bcdots}r_\bcdot\tp_\bcdot/\hbar}}\hspace{.1em}\tvvv\right]=\delta^4(\br),
\end{align*}
When we expand the amphometric around $\theta\!=\!1$, namely the  \textit{quasi-Lorentzian metric}, we obtain that
\begin{align}
[\bm{\eta}^{-1}_\theta]^{ab}&\xrightarrow{\theta\rightarrow1-\delta_0}[\bm{\eta}^{-1}_\delta]^{ab}:=
\textup{diag}(1,-1+i\delta^{~}_0,-1+i\delta^{~}_0,-1+i\delta^{~}_0)+{\cal O}\(\delta_0^2\), ~0<\delta_0\ll1,\label{etaeps}
\intertext{yielding Green's function (\ref{KGprop}) as}
{G_\textsc{F}\!}\(t,\vec{r}\)&:=\left.{G_\theta}\(\delta(\br)\)
\right|_{\theta=1-\delta_0}=
\frac{1}{(2\pi\hbar)^3}
\int\hspace{-.7em}\int\hspace{-.7em}\int_{-\infty}^{+\infty}d\vec{p}\hspace{.2em}
e^{i\vec{r}\cdot\vec{\tp}/\hbar}\(
\lim_{\delta_0\rightarrow+0}\frac{1}{2\pi\hbar}
\int_{-\infty}^{+\infty}d\tilde{E}\hspace{.2em}I_{\hspace{-.2em}\tilde{E}}\),
\intertext{where}
I_{\hspace{-.2em}\tilde{E}}&:=
\frac{{e^{-it\tilde{E}/\hbar}}{e^{t\tilde{E}\delta_0/\hbar}}}
{-\tp^{~}_0\tp^{~}_0+\vec{\tp}\cdot\vec{\tp}+\mu^2-i\delta_0},
\intertext{with}
\btp&:=(\tilde{E},\vtp)=\(\tilde{E},\tp_1,\tp_2,\tp_3\),~
\btr:=(t,\vr)=\(t,r_1,r_2,r_3\).
\end{align}
and  $\vec{r}\cdot\vec{\tp}$ gives an inner product of two three-dimensional spacial vectors owing to the Euclidean metric.
We note that the imaginary part of a denominator is $-i\hspace{.1em}\textrm{sign}[\tilde{E}^2]\hspace{.1em}\delta_0=-i\delta_0$.
A contour integration after extending $\tilde{E}$ to a complex variable $\tilde{E}$ as Re$[\tilde{E}]=\tilde{E}$ provides the $\tilde{E}$ integral such that: 
\begin{align}
\frac{1}{2\pi\hbar}\int_{-\infty}^{+\infty}d\tilde{E}\hspace{.2em}I_{\hspace{-.2em}\tilde{E}}=\frac{I_{\!j}}{2\pi\hbar}
\lim_{\rho\rightarrow\infty}\oint_{c^\rho_j}d\tilde{E}\hspace{.2em}I_{\hspace{-.2em}\tilde{E}}\label{Fint}
\end{align}
with
\begin{align*}
\left\{c^\rho_j,I_{\!j}\right\}=\left\{
\begin{array}{lr}
\left\{c^\rho_1,-1\right\},&t>0\\
\left\{c^\rho_2,+1\right\},&t<0
\end{array}
,\right.
\end{align*}
where contour $c^\rho_1$ ($c^\rho_2$) is along a real axis with $-\rho<\tilde{E}<\rho$ and a lower (upper) half-circle with a radius of $\rho$, respectively, as shown in Figure \ref{feynmen}.
\begin{figure}[t] 
\centering
\includegraphics[width={4cm}]{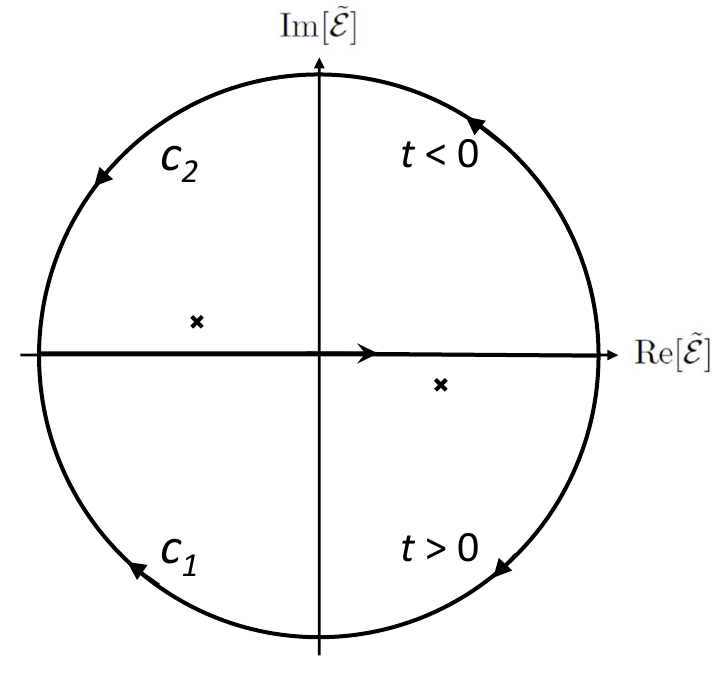}%
 \caption{
The integration contour of the Feynman propagator.
Contour $C_1$ ($C_2$) is taken to (counter-)clockwise.
Two little crosses show a position of poles at $\tilde{E}=\pm\tilde{E}(\vec{\tp}\hspace{.2em},\delta_0)$.
}\label{feynmen}
\end{figure}
Straightforward calculations provide
\begin{align*}
\text{(\ref{Fint})}=&\frac{i\hspace{.1em}e^{-i\left|t\right|
\tilde{E}(\vec{\tp}\hspace{.1em},\delta_0)/\hbar}}
{2\tilde{E}(\vec{\tp}\hspace{.1em},\delta_0)},
\intertext{where}
\tilde{E}(\vec{\tp}\hspace{.2em},\delta_0):=&\sqrt{\vec{\tp}\cdot\vec{\tp}+\mu^2}
-\frac{i\delta_0}{2\sqrt{\vec{\tp}\cdot\vec{\tp}+\mu^2}}+{\cal O}(\delta_0^2)\simeq
\sqrt{\vec{\tp}\cdot\vec{\tp}+\mu^2}-i\delta_0.
\end{align*}
Thus, we obtain Green's function after taking $\delta_0\rightarrow+0$ as
\begin{align*}
{G_\textsc{F}\!}\(t,\vec{r}\):=&\frac{i}{(2\pi\hbar)^3}
\int\hspace{-.7em}\int\hspace{-.7em}\int_{-\infty}^{+\infty}
\frac{d\vec{p}}{2\tilde{E}(\vec{\tp}\hspace{.2em})}
e^{-i\(\left|t\right|\tilde{E}(\vec{\tp}\hspace{.1em})-\vec{r}\cdot\vec{\tp}\)/\hbar},
\end{align*}
where $\tilde{E}(\vec{\tp}\hspace{.1em}):=\tilde{E}(\vec{\tp}\hspace{.1em},1)$; that is referred to as the \emph{Feynman propagator} in physics.
The fundamental solution of the Klein--Gordon equation (\ref{KGeq}) is provided using the Feynman propagator  as
\begin{align*}
\phi(r)=&\phi_0(r)
+\int_{D\(f(r_0)\)}
{G_F}\(t-t_0,\vec{r}-\vec{r}_0\)
f(r_0)\hspace{.1em}d^4r_0,
\end{align*}
where $\phi_0(r)$ is a solution of a homogeneous equation such that $\(\hbar^2\Delta_\theta(r)+\mu^2\)\phi_0(r)=0$. 
\QED
\end{example}
\noindent
We note that the Feynman propagator appears naturally as the propagator in the amphometric space at the $\theta\rightarrow1-0$ limit and is compatible with the causality required from quantum physics.

%
%
\section{Green's functions of gravity}\label{GinG}
\subsection{Linearised vacuum Einstein-equation}
This section discusses Green's function concerning the Einstein equation.
The current author proposed the canonical quantisation of Einstein's gravity in the previous study\cite{doi:10.1140/epjp/s13360-021-01463-3}.
This study provides propagators of spin-connection and vierbein fields.
The common understanding of the geometrical gauge theory is that a connection corresponds to a gauge field, and a curvature corresponds to field strength.
Moreover, a matter field (a Dirac spinor field) in the Yang--Mills theory corresponds to the vierbein form in general relativity since they are sections in $\TsM$.
Therefore, we apply a gauge-fixing condition on the spin-connection field $\omega$.
Although the standard formulation of general relativity does not have a counterpart of a coupling constant in the Yang--Mills theory, this study formulates the $SO_4$ covariant differential with the coupling constant $\cG$ to discuss two theories in parallel.

The Euler--Lagrange equation of motion owing to the Einstein--Hilbert Lagrangian provides the torsion equation and the Einstein equation such that\cite{Kurihara:2022sso,Kurihara:2025tro}:
\begin{align}
&\left(d\eee^a+\cG\hspace{.1em}\www^{a}_{~\bcdot}\wedge\eee^\bcdot\right)=\TTT^a,\label{torsion}
\intertext{and}
\frac{1}{2}\epsilon_{a\bcdots\bcdot}\RRR^\bcdots\wedge\eee^\bcdot&=
\frac{1}{2}\epsilon_{a\bcdots\bcdot}
\(d\www^\bcdots+\cG\www^\bcdot_{~\star}\wedge\www^{\star\bcdot}\)
\wedge\eee^\bcdot
={\eta^{~}_{\theta}}_{a\bcdot}\hspace{.2em}\EEE^{\bcdot}\!.\label{EoM2}
\end{align}
where $\TTT^a$ and  $\EEE^a$ are the torsion two-form and the stress-energy three-form, respectively, owing to matter and gauge fields.
We have provided concrete expressions of the torsion and the stress-energy forms in Ref.\cite{Kurihara:2025tro}.
We set  the cosmological constant to zero for simplicity.
A component representation of (\ref{EoM2}) in the amphometric space is
\begin{align}
R^{ab}-\frac{1}{2}R\hspace{.2em}{\eta^{~}_\theta}^{\hspace{-.1em}ab}&=
\Ri^{\bcdot a}_{\hspace{.7em}\bcdot \star}{\eta^{~}_\theta}^{\hspace{-.1em}b\star}
-\frac{1}{2}\Ri^{\bcdot\star}_{\hspace{.7em}\bcdot\star}{\eta^{~}_\theta}^{\hspace{-.1em}ab}
=\frac{1}{2}\left(\frac{1}{3!}\epsilon^{a\bcdots\bcdot}\left[\EEE^{b}\right]_{\bcdots\bcdot}
+(a{\leftrightarrow}b)\right).\label{EinsteinEq1}
\end{align}
The Bianchi identity (\ref{Bicnchi2}) and the torsion-less (\ref{torsion}) in the vacuum immediately gives that
\begin{align}
\partial_\bcdot\(R^{a\bcdot}-\frac{1}{2}R\hspace{.2em}{\eta^{~}_\theta}^{\hspace{-.1em}a\bcdot}\)&=
\partial^\bcdot\Ri^{a\star}_{\hspace{.7em}\bcdot\star}
-\frac{1}{2}\partial^a\Ri^{\bcdot\star}_{\hspace{.7em}\bcdot\star}=0,\label{EinsteinEq2}
\end{align}
where $\partial^a:={\eta^{~}_\theta}^{\hspace{-.1em}a\bcdot}\partial_\bcdot$.
We used anti-symmetry of $\Ri$ concerning both pairs of upper and lower indexes. 
Owing to the structure equation, we obtain an equation of the spin-connection (Bianchi identity) as
\begin{align}
0=&
\partial^\bcdot\(
\partial_\bcdot\omega_{\star}^{~a\star}-\partial_\star\omega_{\bcdot}^{~a\star}
+\cG
\(
\omega_\bcdot^{~a\bm{\ast}}\omega_\star^{~\bm{\ast}\star}
-\omega_\star^{~a\bm{\ast}}\omega_\bcdot^{~\bm{\ast}\star}\)
{\eta^{~}_{\theta}}_{\bm{\ast}\bm{\ast}}
\)
-\partial^a\(
\partial_\star\omega_\bcdot^{~\star\bcdot}
+\cG\hspace{.1em}\omega_\bcdot^{~\bcdot\bm{\ast}}\omega_\star^{~\bm{\ast}\star}
{\eta^{~}_{\theta}}_{\bm{\ast}\bm{\ast}}
\)\notag\\
=&\(\partial^\bcdot\partial_\bcdot\omega_{\star}^{~a\star}-
\partial^\bcdot\partial_\star\omega_{\bcdot}^{~a\star}-
\partial^a\partial_\star\omega_\bcdot^{~\star\bcdot}\)+
\cG\partial^\bcdot\(
\omega_\bcdot^{~a\bm{\ast}}\omega_\star^{~\bm{\ast}\star}-
\omega_\star^{~a\bm{\ast}}\omega_\bcdot^{~\bm{\ast}\star}-
\delta^a_\bcdot\omega_\star^{~\star\bm{\ast}}\omega_\star^{~\bm{\ast}\star}
\)
{\eta^{~}_{\theta}}_{\bm{\ast}\bm{\ast}}.\label{EinsteinEq3}
\end{align}

In the Yang--Mills theory, one of the standard gauge fixing conditions on a gauge field $A^\bullet$ is the covariant gauge such that $\partial_\mu A^\mu=0$.
Together with the massless condition $g_{\mu\nu}A^\mu A^\nu=0$, two degrees of freedom remain out of four components of the gauge field.
We propose a similar gauge condition on the spin-connection as
\begin{align}
\partial_\bcdot\omega_b^{~a\bcdot}&=0,\label{dwxero}
\end{align}
which provides $16$ constraints on a total $24$ degrees of freedom in $\omega_\mu^{~ab}$.
In addition, the Einstein equation (\ref{EinsteinEq1})  provides six constraints (an on-shell condition); thus, two physical degrees of freedom remain in our theory.

The fundamental solution of the equation of motion with a null coupling constant provides a propagator of the \emph{free} spin-connection field.
Then, the perturbation method concerning the coupling constant $\cG$ provides a more precise solution to the equation of motion.
In reality, an expansion coefficient of a quantum field theoretical scattering matrix is $\alpha_{\hspace{-.1em}g\hspace{-.1em}r}:={{c^{2}_{\hspace{-.1em}g\hspace{-.1em}r}}}/4\pi$\footnote{
A coupling constant in the Yang-Mills theory has a physical dimension (e.g., C (Coulomb) for the electric charge), and the dimensionless expansion coefficient is $\alpha=e^2/(4\pi\epsilon_0\hbar c)$ in SI.
On the other hand, gravitational coupling $\cG$ itself is dimensionless.
}.
For the standard value of $\cG=1$, the expansion coefficient for gravity is $\alpha_{\hspace{-.1em}g\hspace{-.1em}r}\approx0.0796$, which is smaller than strong interaction ($\alpha_s\approx0.1$), but larger than electromagnetic interaction ($\alpha\approx0.00730$). 
Equation (\ref{EinsteinEq3}) under the \emph{free} spin-connection field approximation with gauge fixing and massless conditions is the simple Laplace equation such that:
\begin{align}
\(\Delta_\theta[\bm\omega^{~}_0]_{\bcdot}^{~a\bcdot}\)\(\xi\in\TMg\)&=0,\label{deltaomega}
\end{align}
$\bm\omega^{~}_0$ is the \emph{free} spin-connection field fulfilling gauge fixing and massless conditions.
The fundamental solution of (\ref{deltaomega}) with quasi-Minkowski metric (\ref{etaeps}) is
\begin{align}
{G^a_{\omega}}(\bm{\xi}):=[\bm\omega^{~}_0(\bm{\xi})]_{\bcdot}^{~a\bcdot}=\frac{\tilde\omega^a}
{-(\xi^0)^2+(\xi^1)^2+(\xi^2)^2+(\xi^3)^2-i\delta_0},
\end{align}
owing to (\ref{uAlpha}) for a massless particle ($\mu=0$), where $\tilde\omega^a$ is a polarization vector of the spin-connection field, which are two independent constant vectors.
The spin-connection field is a rank-two massless tensor field corresponding to a spin-two boson. 

Next, we consider Green's function of the vierbein field.
The torsion two-form has a component representation in $\TsM$ and $\TsMM$ using the standard bases, such as  
\begin{align*}
\TTT^a&=:\frac{1}{2}\TT^a_{\hspace{.6em}\bcdots}\hspace{.2em}\eee^\bcdot\wedge\eee^\bcdot=
\frac{1}{2}\TT^a_{\hspace{.6em}\bcdots}\hspace{.2em}\E^\bcdot_\mu\E^\bcdot_\nu
\hspace{.2em}dx^\mu{\wedge}dx^\nu\!.
\end{align*}
Torsion equation (\ref{torsion}) is represented using this component representation such as
\begin{align*}
\TT^a_{\hspace{.6em}bc}&=\(\partial^{~}_\mu\E^a_\nu-\partial^{~}_\nu\E^a_\mu+
\cG\(\omega_\mu^{~a\bcdot}\E^\bcdot_\nu-
\omega_\nu^{~a\bcdot}\E^\bcdot_\mu\)
{\eta^{~}_{\theta}}_{\bcdots}\)
\E_b^\mu\E_c^\nu,\notag\\&=\E^a_\mu
\(\partial^{~}_c\E_b^\mu-
  \partial^{~}_b\E_c^\mu\)+
\cG\({\omega_b^{~a}}_c-{\omega_c^{~a}}_b\)=0,
\end{align*}
where ${\omega_b^{~a}}_c:=\omega_b^{~a\bcdot}{\eta^{~}_{\theta}}_{\bcdot c}$.

Divergence-less concerning an upper index of the torsion form and gauge fixing condition (\ref{dwxero}) gives constraints such as
\begin{align*}
\partial^{~}_\bcdot\TT^\bcdot_{\hspace{.6em}bc}&=
\partial_\mu\(\partial^{~}_c\E_b^\mu-
  \partial^{~}_b\E_c^\mu\)
+
\(\partial^{~}_c\E_b^\mu-
  \partial^{~}_b\E_c^\mu\)
\E_a^\nu\partial_\nu\E^a_\mu,\notag\\
&=\partial^{~}_c\partial^{~}_\mu\E^\mu_b-\partial^{~}_b\partial^{~}_\mu\E^\mu_c
-\E^a_\nu\(\partial^{~}_c\E_b^\nu-
  \partial^{~}_b\E_c^\nu\)
\partial_\mu\E_a^\mu.
\end{align*}
Gauge fixing condition  (\ref{dwxero}) yields constraints equivalent to the de\hspace{.1em}Donder condition on the vierbein such that:
\begin{align}
\partial^{~}_\bcdot\TT^\bcdot_{\hspace{.6em}bc}=0\land\text{(\ref{dwxero})}\implies( \partial^{~}_\mu\E^\mu_\bullet)=0.\label{deDonder2}
\end{align}
On the other hand, divergence concerning a lower index provides an equation of motion for the vierbein with the constraint (\ref{deDonder2}) as
\begin{align*}
0&=g_{~}^{\hspace{.2em}\mu\nu}\partial^{~}_{\mu}\TT^a_{\hspace{.6em}\nu\rho}\\&=
\Delta(x)\E^a_\rho-g_{~}^{\hspace{.2em}\mu\nu}\partial^{~}_\rho\partial^{~}_{\mu}\E^a_\nu\\&\hspace{5.1em}
+\cG\hspace{.2em}g_{~}^{\hspace{.2em}\mu\nu}\partial_\mu^{~}
\(\omega_\nu^{~a\bcdot}\E^\bcdot_\rho-
\omega_\rho^{~a\bcdot}\E^\bcdot_\nu\)\eta_{\bcdots},\\
&\xrightarrow{~\cG=0~}\Delta(x)\left[\E^0\right]^a_\nu=0;
\end{align*}
thus, an equation of motion under the \emph{free} vierbein-field approximation is again the simple Laplace equation such that:
\begin{align}
\(\Delta\left[\E^0\right]^a_\nu\)(x\in\TMM)=0.\label{deltaE}
\end{align}
The fundamental solution of (\ref{deltaE}) with the quasi-Minkowski metric is
\begin{align*}
{G^a_{\E}}(x):=\left[\bm\E^0\right]_{\mu}^{a}=\frac{{\E}_{\mu}^{a}}
{-(x^0)^2+(x^1)^2+(x^2)^2+(x^3)^2-i\delta_0},
\end{align*}
where ${\E}_{\mu}^{(a)}$ is a polarization vector of the vierbein field.
The vierbein field is a massless tensor field; the upper Roman index is a component in the inertial manifold, and the lower Greek index is a component in the global manifold.
We interpret the vierbein field as the global vector equipping the local $SO_4$ gauge symmetry.
The lower index represents a vector component concerning global space-time. The upper index points to an inner degree of freedom corresponding to two polarizations of the local $SO_4$ gauge symmetry.
Under this interpretation, we represent the vierbein one-form defined in $\TsMM$ as
\begin{align*}
{\E}_{\mu}^{a}\rightarrow{\E}_{\mu}^{(a)}\implies
{\eee}^{(a)}:={\E}^{(a)}_\mu\hspace{-.1em}(\bm{x})\hspace{.2em}dx^\mu,
\end{align*}
where $a=1,2$ shows the inner degree of freedom two as independent polarization vectors.
We can construct a spin-two polarization state of the metric tensor from the vierbein field using the coupling of angular momenta such as $|2,\pm2\rangle=|\pm1\rangle\otimes|\pm1\rangle$.
The vierbein has no spin-zero state of $|\pm1,0\rangle$ since it is massless; thus, the metric tensor does not have $|2,\pm1\rangle$ nor $|2,0\rangle$.

The metric tensor is a symmetric tensor with ten degrees of freedom in four-dimension.
On the other hand, $\E^a_\mu$ is not necessarily a symmetric tensor and has 16 degrees of freedom in the four-dimensional space.
A local $SO_4$ symmetry provides six constraints on the vierbein field; thus, the same number of degrees of freedom as the metric tensor remains in the vierbein field.
In addition, four gauge conditions (\ref{deDonder2}) and four torsion-less conditions provide eight additional constraints.
Consequently, the vierbein field has two physical degrees of freedom, as we desired.

%
%
\begin{example}(A plane wave solution on a flat space-time)\label{pw}\\
We consider a polarization vector for a plane wave solution in flat space-time.
A solution propagating along, e.g.,  an $x^1$-axis, has a polarization tensor such that:
\begin{align*}
&{\E}_{\mu}^{(a)}:=\frac{1}{\sqrt{2}}
\left(
\begin{array}{cccr}
1&i&0&0\\
i&1&0&0\\
0&0&1&1\\
0&0&1&-1
\end{array}
\right)\\&\implies
{g^{~}_{\theta}}_{\mu\nu}={\eta^{~}_\theta}_\bcdots{\E}_{\mu}^{(\bcdot)}{\E}_{\nu}^{(\bcdot)}=
\textup{diag}\(1,e^{i\pi\theta},e^{i\pi\theta},e^{i\pi\theta}\),
\end{align*}
with a \emph{circular} polarization.
This polarization vector yields the spin-connection field such that: 
\begin{align*}
\omega^{(a)}=\frac{1}{\sqrt{2}}
\(0,0,\omega_{2}^{~23}-\omega_{3}^{~23},-\omega_{2}^{~23}-\omega_{3}^{~23}\).
\end{align*}
We note that states with two polarization vectors, 
\begin{align*}
\eee^{(\lambda=\bm{+})}&:={\E}_{\mu}^{(2)}dx^\mu=\frac{1}{\sqrt{2}}\(dx^2+dx^3\)
\intertext{and }
\eee^{(\lambda=\bm{\times})}&:={\E}_{\mu}^{(3)}dx^\mu=\frac{1}{\sqrt{2}}\(dx^2-dx^3\),
\end{align*}
have a dynamic degree and are physically observable in the quantum field theory, similar to the polarization vector of photon in QED.
The upper index of the vierbein corresponds to an index to specify two transverse, one longitudinal and one scalar photons of the photon polarization vector.
On the other hand, two polarization vectors of the spin-connection field have no counterparts in QED and are two independent vector functions:
\begin{align*}
\www^{(\lambda=1)}&:=\omega_{2}^{(1)}\hspace{.1em}\eee^{2}=
\omega_{\mu}^{(1)}[\E^{-1}]^\mu_2\hspace{.1em}dx^2=
\omega^{(1)}dx^2
\intertext{and }
\www^{(\lambda=2)}&:=\omega_{3}^{(2)}\hspace{.1em}\eee^3=
\omega_{\mu}^{(2)}[\E^{-1}]^\mu_3\hspace{.1em}dx^3=
\omega^{(2)}dx^3,
\intertext{where }
\omega^{(1)}&:=\frac{1}{{\sqrt{2}}}\(\omega_{2}^{(1)}+\omega_{3}^{(1)}\)
\intertext{and}
\omega^{(2)}&:=\frac{1}{{\sqrt{2}}}\(\omega_{2}^{(2)}+\omega_{3}^{(2)}\).
\end{align*} 
Consequently, we obtain Green's functions in the momentum space such that
\begin{align*}
{\widetilde{G}^{(\lambda)}_\E}(\btp)&:=\frac{\tilde\E_a^{(\lambda)}}
{-\tp^{~}_0\tp^{~}_0+\vec{\tp}\cdot\vec{\tp}-i\delta_0},\hspace{2em}\lambda=(\bm{+},\bm{\times}),
\intertext{and}
\widetilde{G}^{(\lambda)}_\omega(\btp)&:=\frac{\tilde\omega^{(\lambda)}}
{-\tp^{~}_0\tp^{~}_0+\vec{\tp}\cdot\vec{\tp}-i\delta_0},\hspace{2em}\lambda=(1,2),
\end{align*}
in the quasi-Minkowski metric.
Polarization vectors $\tilde\E_a^{(\lambda)}$ are components of the vector defined in $\TstMg$ such that: 
\begin{align*}
\tilde\eee^{+}&=\tilde{\E}_{\bcdot}^{(+)}\hspace{.1em}d\tp^\bcdot=\frac{1}{\sqrt{2}}
\tilde{\E}(\btp)\(d\tp^2+d\tp^3\)
\intertext{and }
\tilde\eee^{\times}&=\tilde{\E}_{\bcdot}^{(\times)}\hspace{.1em}d\tp^\bcdot=\frac{1}{\sqrt{2}}
\tilde{\E}(\btp)\(d\tp^2-d\tp^3\),
\end{align*}
where $\tilde{\E}$ is the $C\hspace{-.2em}M$-Fourier transformation of unity.
Similarly, $\tilde\omega^{(\lambda)}(\btp)$ are given as the $C\hspace{-.2em}M$-Fourier transformation of $\omega^{(\lambda)}(\bm{x})$.
As  a result, $\hat{G}^{(\lambda)}_\bullet$ are Green's functions  in a momentum space as a distribution\cite{talagrand_2022}.
\QED
\end{example}

\subsection{Green's function owing to exact solutions}
This section treats vierbein and spin-connection fields as classical background fields in the quantum field theory. 
In the Yang--Mills theory with $SU(N)$ gauge group in curved space-time, a classical equation of motion of spinor field $\psi$ is   
\begin{align*}Xre
&~\hbar\left(i\gamma^\bcdot\partial_\bcdot
-\frac{1}{2}\cG\hspace{.1em}\gamma^\bcdot\E_\bcdot^\mu\omega_\mu^{~\stars}{S}_\stars
+{ic^{~}_\SU}\hspace{.1em}\gamma^\bcdot\Aa^I_\bcdot t_I
\right)\psi={\mu}\hspace{.1em}\psi,~\text{with}~
S^{ab}:=\frac{i}{4}[\gamma^a,\gamma^b],
\end{align*}
where $\Aa^I_a$ are the connection of the gauge field  with coupling constant ${c^{~}_\SU}$, $\mu$ is a mass of the spinor field, $\gamma^\bullet$ are  Dirac matrices, $S^\bullets$ are generators of a $spin(1,3)$ group, and  $t_I$ is the generator of $SU\!(\!N\!)$-gauge group.
We note that all fields $\psi$, $\E\omega$, $\Aa$ and differential operator $\partial_a=\partial/\partial{\xi^a}$ are defined in $\TMg$. 
A definition of the Dirac gamma matrices in the amphometric space is given in Ref.\cite{Kurihara:2025tro}.
In the quantum field theory, the standard perturbation calculation method (Feynman-diagram method) treats all fields in the momentum space.
We can utilize the standard method of the quantum Yang--Mills theory in classical curved background space-time when we have vierbein and spin-connection fields as exact solutions of the Einstein equation in the momentum space.  

In a semi-classical theory with a gravitational background, the quantum field theory treats a gravitational field as a classical external field using the exact solution of the Einstein equation.
Vierbein and spin-connection fields in a configuration space are Fourier-transformed to Green's function in momentum space and interact with quantum matter and gauge fields.
We consider Green's functions of vierbein and spin-connection fields of the Schwarzschild solution:
%
%
\begin{figure*}[t] 
\centering
\includegraphics[width={15cm}]{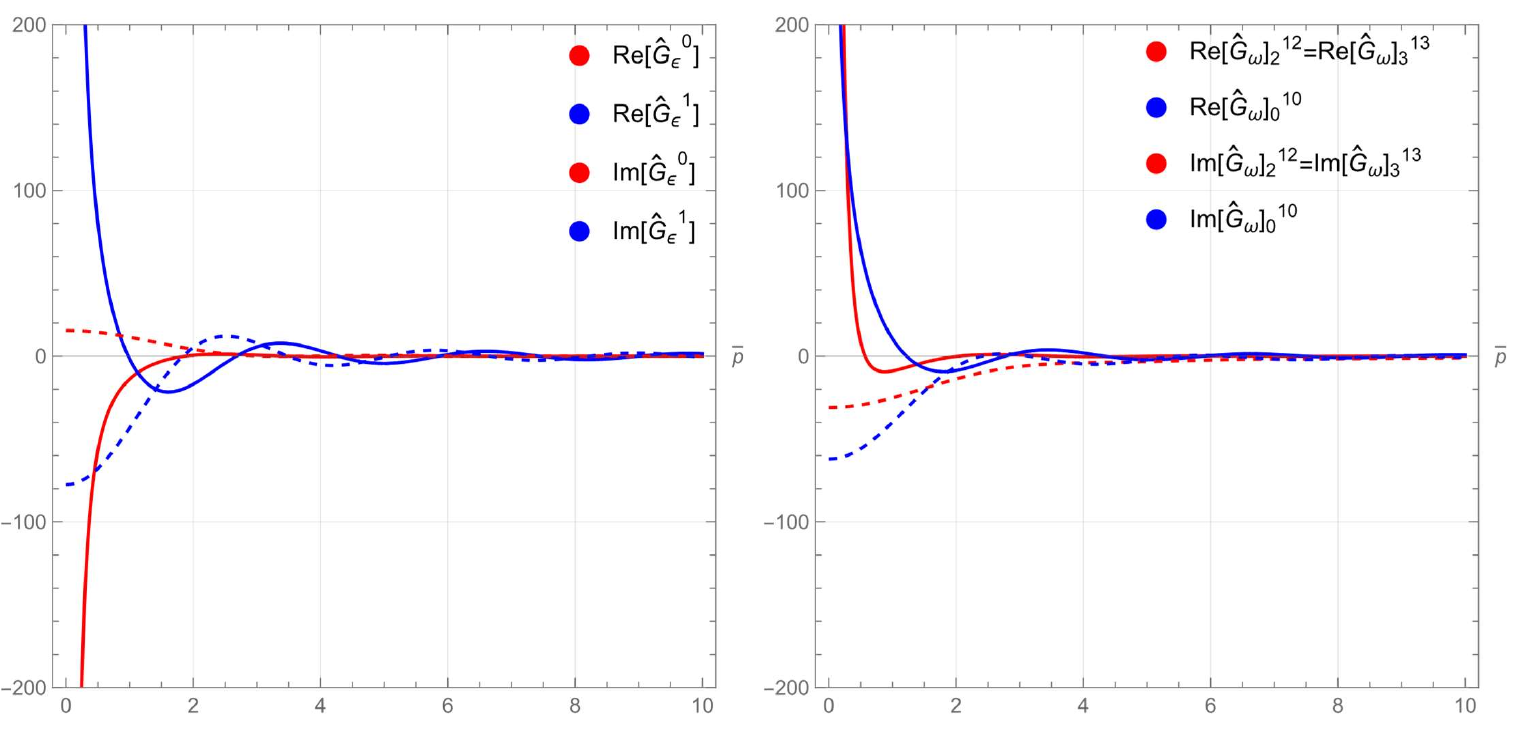}%
 \caption{
Green's functions of the vierbein (left) and the spin-connection (right) in the momentum space for the Schwarzschild solution
as a function of $\bar{p}$.
Numerical values are normalised by a factor $(2\pi\hbar)^2$ with $2\Rs=\delta{E}=1$.
}
\label{Greens}
\end{figure*}
\begin{example}(Schwarzschild solution)\label{schw}
The vierbein forms of the Schwarzschild solution using polar-coordinate $d\rho^\mu=(dt,dr,d\vartheta,d\varphi)$ in $\TsMg$ are 
\begin{subequations}
\begin{align}
&\eee^a=\E^a_\mu\hspace{.1em}dx^\mu,\notag
\intertext{where}
&\left\{
\begin{array}{clcl}
\eee^0&=\fsch(r)\hspace{.1em}dt,&\eee^1&=\fsch(r)^{-1}\hspace{.1em}dr,\\
\eee^2&=r\hspace{.1em}d\vartheta,&\eee^3&=r\sinvt\hspace{.1em}d\varphi,
\end{array}
\right.
\label{SchwrzE}
\intertext{and}
&\fsch(r):=\sqrt{1-\frac{\Rs}{r}},~~~
\Rs:=\frac{\kappa M}{4\pi}\label{fr},
\end{align} 
where $\Rs$ is the Schwarzschild radius.
\end{subequations}
Vierbein matrices providing vierbein forms (\ref{SchwrzE}) in the amphometric space are 
\begin{align*}
\E^a_\mu{dx^\mu}&=[\E_\textrm{Schw}]^a_\bcdot[\E_{pol}]^\bcdot_\mu{dx^\mu}\\
&=\left(
\begin{array}{cccc}
\fsch(r)&0&0&~0\\
0&\fsch(r)^{-1}&0&~0\\
0&0&1&~0\\
0&0&0&~1
\end{array}
\right)
\left(
\begin{array}{crrr}
e^{i\pi\theta}&0~~~~&0~~~~&0~~\\
0&\cos{\varphi}\sin{\vartheta}&\sin{\varphi}\sin{\vartheta}&\cos{\vartheta}\\
0&r\cos{\varphi}\cos{\vartheta}&r\sin{\varphi}\cos{\vartheta}&-r\sin{\vartheta}\\
0&-r\sin{\varphi}\sin{\vartheta}&r\cos{\varphi}\sin{\vartheta}&0~~
\end{array}
\right)\left(
\begin{array}{c}
dt\\dx\\dy\\dz
\end{array}
\right).
\end{align*}
We consider Green's function of the vierbein  $\E_\mu^a(\bm{x})$: a $C\hspace{-.2em}M$-Fourier transformation of Green's function in the configuration space  provides that  in the momentum space.
Green's function with the coordinate vector $dp^a=\(dE,dp,d\vartheta,d\varphi\)$ is 
\begin{align*}
\widetilde{G}^a_\E(E_p,|\vec{\tp}|)=&
\int_{-\infty}^{\infty}dt\hspace{.1em}e^{i(E_p+i\delta_0)|t|/\hbar}
\int_0^{2\pi}d\varphi\int_{-1}^1d\hspace{-.1em}\cos{\hspace{-.1em}\vartheta}\int_0^{\infty}dr\hspace{.2em}
e^{-i\vec{r}\cdot\vec{\tp}/\hbar}
\hspace{.5em}r^2\(\fsch(r),\fsch(r)^{-1},1,1\),
\end{align*}
where $\vec{r}\cdot\vec{p}=p\hspace{.2em}r\cos{\vartheta}$, $0<\delta_0\ll\hspace{.1em}E_p)$ and $E_p>0$ is an energy of a \emph{graviton}.
We obtain Green's functions of the vierbeins after the Fourier transformation as
\begin{align*}
\frac{\widetilde{G}^0_\E\(E_p,\bar{p}\)}{2\pi\hbar}&=
\frac{{\pi^2\Rs^3}}{4}\(\frac{
-\cos{\hspace{-.1em}\bar{p}}\hspace{.2em}J_1\hspace{-.2em}\(\bar{p}\)+
\sin{\hspace{-.1em}\bar{p}}\hspace{.2em}Y_1\hspace{-.2em}\(\bar{p}\)}
{\bar{p}^2}-
2i\frac{\sin{\hspace{-.1em}\hspace{.1em}\bar{p}}\hspace{.2em}J_1\hspace{-.2em}\(\bar{p}\)}{\bar{p}^2}\)
\hspace{.1em}\delta\hspace{-.1em}E_p,\\
&+-\frac{\pi\Rs^3}{2}
\(\frac{1}{\bar{p}^{2}}+\frac{\pi}{4\bar{p}}-\frac{1}{2}
\(\log{\frac{\bar{p}}{2}}-\frac{1}{6}+\gamma^{~}_E+i\pi\)
+{\cal O}\(\bar{p}^{1}\)\)\hspace{.1em}\delta\hspace{-.1em}E_p,\\
\frac{\widetilde{G}^1_\E\(E_p,\bar{p}\)}{2\pi\hbar}&=
\frac{\pi^2\Rs^3}{4}\left\{
\cos{\hspace{-.1em}\bar{p}}\hspace{.1em}\(2\frac{J_0(\bar{p})+Y_1(\bar{p})}{\bar{p}}-\frac{J_1(\bar{p})}{\bar{p}^2}\)-
\sin{\hspace{-.1em}\bar{p}}\(2\frac{J_1(\bar{p})+Y_0(\bar{p})}{\bar{p}}-\frac{Y_1(\bar{p})}{\bar{p}^2}\)
\right.\\&\left.\hspace{1em}
+2i\(\sin{\hspace{-.1em}\bar{p}}\(\frac{J_1(\bar{p})}{\bar{p}^2}-\frac{2J_0(\bar{p})}{\bar{p}}\)-
\cos{\hspace{-.1em}\bar{p}}\hspace{.1em}\hspace{.2em}\frac{2J_1(\bar{p})}{\bar{p}}\)
\right\}
\hspace{.1em}\delta\hspace{-.1em}E_p,\\
&=\frac{\pi\Rs^3}{2}
\(
\frac{1}{\bar{p}^{2}}+\frac{3\pi}{4\bar{p}}-\frac{5}{2}
\(\log{\frac{\bar{p}}{2}}+\gamma^{~}_E-+i\pi+\frac{7}{30}\)
+{\cal O}\(\bar{p}^{1}\)\)\hspace{.1em}\delta\hspace{-.1em}E_p,,\\
\frac{\widetilde{G}^2_\E\(E_p,\bar{p}\)}{2\pi\hbar}&=\frac{\widetilde{G}^3_\E\(E_p,\bar{p}\)}{2\pi\hbar}=
\pi\hbar\hspace{.2em}\delta^2{\hspace{-.2em}E\bar{p}},
\intertext{where}
\bar{p}&:=\frac{{\Rs}}{2\hbar}|\vec{\tp}|,~~
\delta^2{\hspace{-.2em}E\bar{p}}
=2\pi\hbar^2\frac{\delta(|\vec{\tp}|)}{|\vec{\tp}|^2}\delta{E_p}
=\frac{\pi\Rs^3}{4\hbar}\frac{\delta(\bar{p})}{\bar{p}^2}\delta{E_p},~~
\delta{E_p}:=\left.\hspace{.1em}\delta(E_p+i\delta_0)\right|_{\delta_0\rightarrow0},
\end{align*}
$ \gamma^{~}_E$ is a Euler's constant.
$J_i(\bullet)$ and $Y_i(\bullet)$ $(i=1,2)$ are Bessel functions of the first and the second kind, respectively.
We note that $\bar{p}$ is dimensionless momentum variable.

Corresponding spin-connection forms 
$
[\omega\E]_c^{~ab}:=\cG\hspace{.2em}\omega_\mu^{~ab}[\E^{-1}]^\mu_c
$
 are 
\begin{align*}
[\omega\E]^{~01}_0&=-[\omega\E]^{~10}_0=-\frac{{\Rs}}{2\fsch(r)r^2},\hspace{2em}
[\omega\E]^{~12}_2=-[\omega\E]^{~21}_2=\frac{\fsch(r)}{r},\\
[\omega\E]^{~13}_3&=-[\omega\E]^{~31}_3=\frac{\fsch(r)}{r},\hspace{2.9em}
[\omega\E]^{~23}_3=-[\omega\E]^{~32}_3=\frac{\cos\vartheta}{r\sin{\vartheta}},
\end{align*}
and otherwise zero. 

A Fourier transformation provides Green's function in the momentum space such that:
\begin{align*}
\frac{\left[\widetilde{G}_{\omega}(E_p,\bar{p})\right]^{~10}_0}{2\pi\hbar}&=-\frac{\pi^2\Rs^2}{2}\frac{
\sin{\hspace{-.1em}\bar{p}}\hspace{.2em}Y_0(\bar{p})-
\(\cos{\hspace{-.1em}\bar{p}}\hspace{.1em}-2i\sin{\hspace{-.1em}\bar{p}}\)J_0(\bar{p})
}{\bar{p}}\hspace{.1em}\delta\hspace{-.1em}E_p,\\
&=\frac{\pi\Rs^2}{2}\(
\frac{\pi}{\bar{p}}-2\(\log{\frac{\bar{p}}{2}}+\gamma^{~}_E+i\pi\)+{\cal O}\(\bar{p}^1\)
\)\hspace{.1em}\delta\hspace{-.1em}E_p,\\
\frac{\left[\widetilde{G}_{\omega}(E_p,\bar{p})\right]^{~12}_2}{2\pi\hbar}&=
\frac{\left[\widetilde{G}_{\omega}(E_p,\bar{p})\right]^{~13}_3}{2\pi\hbar},\\
&=\frac{\pi^2\Rs^2}{2}\frac{
\sin{\hspace{-.1em}\bar{p}}\(Y_0(\bar{p})-J_1(\bar{p})-2iJ_0(\bar{p})\)-
\cos{\hspace{-.1em}\bar{p}}\hspace{.1em}\(Y_1(\bar{p})+J_0(\bar{p})-2iJ_1(\bar{p})\)
}{\bar{p}}\hspace{.1em}\delta\hspace{-.1em}E_p,\\
&=
\frac{\pi\Rs^2}{2}\(
\frac{2}{\bar{p}^2}-\frac{\pi}{\bar{p}}
+\log{\frac{\bar{p}}{2}}+\gamma^{~}_E-\frac{1}{2}-i\pi
+{\cal O}\(\bar{p}^1\)\)\hspace{.1em}\delta\hspace{-.1em}E_p,\\
\frac{\left[\widetilde{G}_{\omega}(E_p,\bar{p})\right]^{~23}_3}{2\pi\hbar}&=0.
\end{align*}
A behavior of Green's functions with $\delta{E}=1$ are drawn in Figure \ref{Greens}.
An asymptotic behavior of Green's functions is $\widetilde{G}_{\E}(1,\infty)=\widetilde{G}_{\omega}(1,\infty)=0$.
\QED
\end{example}

%
%
\section{Summary}\label{Summary}
The author proposed a non-perturbative canonical quantisation of general relativity in the Heisenberg picture in the previous work\cite{doi:10.1140/epjp/s13360-021-01463-3}.
For concrete quantum mechanical calculations of gravitational phenomena, we need a perturbative expansion of physical processes, such as scattering processes concerning gravitons.
This report provided wave functions and Green's functions (propagators) of gravitational fields of the vierbein and the spin-connection fields.
They are essential ingredients of perturbed calculations using the diagrammatic method.
The existence of Green’s function for the Laplace--Beltrami operator in curved space and with an indefinite metric is ensured owing to the Hodge harmonic analysis extended to the amphometric space.
The analyticity of Green's function is naturally fixed by the amphometric, which is consistent with keeping a causality as expected from physics.

The standard method of a perturbative quantum field theory in a flat Minkowski space utilizes the momentum space for calculations; the momentum space consists of vectors Fourier transferred from the configurations space.
An integration kernel of the Fourier transformation is a solution of the wave equation; thus, the Fourier transformation corresponds to an expansion of solutions of the exact (interacting) non-linear equation of motion by the fundamental solutions of free (non-interacting) fields.
In curved space-time, the physical interpretation of the momentum space is not straightforward, especially for gravitational fields of the vierbein and spin-connection forms.
We note that the general relativity standard linearization method with the weak-field approximation differs from that in the quantum Yang-Mills theory. 

This report proposed a novel definition of the momentum space in curved space-time.
A vector in the configuration space is transferred to that in the momentum space through two steps, the Fourier(-Laplace)-transformation and the $C\hspace{-.2em}M$ transformation.
The spin-connection field owns the information on the curved space-time in the configuration space. The $C\hspace{-.2em}M$-Fourier(--Laplace) transformed vierbein, and spin-connection fields are defined in the momentum space with the flat metric.
This report also proposed linearising the Einstein equation as a free field consistent with that for the Yang-Mills gauge field.  
The proposed linearisation does not utilize the weak-field approximation; thus, the method applies to highly caved space-time.

We gave two examples of exact Green's functions of gravitational fields: the plane wave solution and the Schwarzschild solution.
After the perturbative quantisation of gravitational fields, we can calculate the gravitational effects using the diagrammatic method to the quantum Yang--Mills theory owing to Green's functions in this report. 
For example, the Hawking radiation from black holes and a gravitational correction of the muon (and also electron) anomalous magnetic moment are possible applications.
 
%
%
\vskip 3mm
\section*{Acknowledgements}
I would like to thank Dr Y$.$ Sugiyama, Prof$.$ J$.$ Fujimoto and Prof$.$ T$.$ Ueda for their continuous encouragement and fruitful discussions.

\bibliographystyle{unsrt} 
\bibliography{Fourier-GR-PRD}      

\begin{thebibliography}{10}

\bibitem{PhysRev.101.1597}
Ryoyu Utiyama.
\newblock Invariant theoretical interpretation of interaction.
\newblock {\em Phys. Rev.}, 101:1597--1607, Mar 1956.

\bibitem{1974Natur.248...30H}
S.~W. {Hawking}.
\newblock {Black hole explosions?}
\newblock {\em Nature}, 248(5443):30--31, March 1974.

\bibitem{PhysRevD.14.870}
W.~G. Unruh.
\newblock Notes on black-hole evaporation.
\newblock {\em Phys. Rev. D}, 14:870--892, Aug 1976.

\bibitem{Birrell:1982ix}
N.~D. Birrell and P.~C.~W. Davies.
\newblock {\em {Quantum Fields in Curved Space}}.
\newblock Cambridge Monographs on Mathematical Physics. Cambridge Univ. Press,
  Cambridge, UK, 2 1984.

\bibitem{Wald:1995yp}
Robert~M. Wald.
\newblock {\em {Quantum Field Theory in Curved Space-Time and Black Hole
  Thermodynamics}}.
\newblock Chicago Lectures in Physics. University of Chicago Press, Chicago,
  IL, 1995.

\bibitem{Ford:1997hb}
L.~H. Ford.
\newblock {Quantum field theory in curved space-time}.
\newblock In {\em {9th Jorge Andre Swieca Summer School: Particles and
  Fields}}, pages 345--388, 7 1997.

\bibitem{mukhanov2007introduction}
V.~Mukhanov and S.~Winitzki.
\newblock {\em Introduction to Quantum Effects in Gravity}.
\newblock Cambridge University Press, Cambridge, UK, 2007.

\bibitem{61354}
Christopher~J. Fewster.
\newblock {Lectures on quantum field theory in curved spacetime}.
\newblock {Preprint}, 2008.

\bibitem{Parker:2009uva}
Leonard~E. Parker and D.~Toms.
\newblock {\em {Quantum Field Theory in Curved Spacetime}: {Quantized Field and
  Gravity}}.
\newblock Cambridge Monographs on Mathematical Physics. Cambridge University
  Press, Cambridge, UK, 8 2009.

\bibitem{PhysRev.96.191}
C.~N. Yang and R.~L. Mills.
\newblock Conservation of isotopic spin and isotopic gauge invariance.
\newblock {\em Phys. Rev.}, 96:191--195, Oct 1954.

\bibitem{Kurihara_2020}
Yoshimasa Kurihara.
\newblock {Symplectic structure for general relativity and
  Einstein--Brillouin--Keller quantization}.
\newblock {\em Classical and Quantum Gravity}, 37(23):235003, oct 2020.

\bibitem{Kurihara:2025tro}
Yoshimasa Kurihara.
\newblock {Yang-Mills-Utiyama Theory and Graviweak Correspondence}.
\newblock 1 2025.

\bibitem{doi:10.1140/epjp/s13360-021-01463-3}
Yoshimasa Kurihara.
\newblock {Nakanishi--Kugo--Ojima quantization of general relativity in
  Heisenberg picture}.
\newblock {\em Eur. Phys. J. Plus}, 136(4):462, 2021.

\bibitem{fre2012gravity}
Pietro~Giuseppe Fr{\`e}.
\newblock {\em Gravity, a Geometrical Course}, volume 1: Development of the
  Theory and Basic Physical Applications.
\newblock 2013.

\bibitem{yosida2012functional}
K.~Yosida.
\newblock {\em Functional Analysis}.
\newblock Classics in Mathematics. Springer, Berlin Heidelberg, 2012.

\bibitem{warner1983foundations}
F.W. Warner.
\newblock {\em Foundations of Differentiable Manifolds and Lie Groups}.
\newblock Graduate Texts in Mathematics. Springer, Berlin, Heidelberg, 1983.

\bibitem{Bernstein1910}
Serge Bernstein.
\newblock Sur la généralisation du problème de dirichlet. (deuxieme partie).
\newblock {\em Mathematische Annalen}, 69:82--136, 1910.

\bibitem{sauvigny2006partial}
F.~Sauvigny.
\newblock {\em Partial differential equations}.
\newblock Springer, Heidelberg, 2006.

\bibitem{krylov2008lectures}
N.V. Krylov.
\newblock {\em Lectures on Elliptic and Parabolic Equations in Sobolev Spaces}.
\newblock Graduate Studies in Mathematics, Graduate Studies in Mathema.
  American Mathematical Society, Providence, 2008.

\bibitem{talagrand_2022}
Michel Talagrand.
\newblock {\em What Is a Quantum Field Theory?}
\newblock Cambridge University Press, Cambrifge, 2022.

\bibitem{Kurihara:2022sso}
Yoshimasa Kurihara.
\newblock {Topological indices of general relativity and Yang--Mills theory in
  four-dimensional space-time}.
\newblock 5 2022.

\end{thebibliography}
\end{document}